\documentclass[lettersize,journal,10pt]{IEEEtran}
\usepackage{amsmath,amsfonts}
\usepackage{algorithmic}
\usepackage{algorithm}
\usepackage{array}
\usepackage{textcomp}
\usepackage{stfloats}
\usepackage{url}
\usepackage{verbatim}
\usepackage{graphicx}
\usepackage{cite}
\usepackage[nolist,nohyperlinks]{acronym}
\begin{acronym}
    \acro{AP}{access point}
    \acro{CCDF}{complementary cumulative density function}
    \acro{CDF}{cumulative density function}
    \acro{PDF}{probability density function}
    \acro{PPP}{Poisson point process}
    \acro{BS}{base station}
    \acro{LOS}{line of sight}
    \acro{SINR}{signal-to-interference-plus-noise ratio}
    \acro{SIR}{signal-to-interference ratio}
    \acro{SNR}{signal-to-noise ratio}
    \acro{PV}{Poisson-Voronoi}
    \acro{LOS}{line-of-sight}
    \acro{NLOS}{non line-of-sight}
    \acro{PGFL}{probability generating functional}
    \acro{PLCP}{Poisson line Cox process}
    \acro{MBS}{macro base station}
    \acro{PLT}{Poisson line tessellation}
    \acro{SBS}{small cell base station}
    \acro{RAT}{radio access technique}
    \acro{RCS}{radar cross section}
    \acro{5G}{fifth generation}
    \acro{PLP}{Poisson line process}
    \acro{UE}{user equipment}
    \acro{BPP}{binomial point process}
    \acro{BLP}{binomial line process}
    \acro{BLCP}{binomial line Cox process}
    \acro{RSU}{road-side unit}
    \acro{SG}{stochastic geometry}
    \acro{MD}{meta distribution}
    \acro{SF}{signal fraction}
    \acro{PCP}{Poisson cluster process}
    \acro{UCB}{Upper Confidence Bound}
    \acro{TS}{Thompson Sampling}
    \acro{CSP}{conditional success probability}
    \acro{MAB}{multi-armed bandit}
    \acro{QoS}{Quality of Service}
    \acro{FMCW}{frequency modulated continuous wave}
    \acro{V2X}{vehicle-to-everything}
    \acro{MAC}{medium access control}
    \acro{TDMA}{time division multiple access}
    \acro{MHCP}{Matern hardcore point process}
    \acro{UAV}{unmanned aerial vehicle}
\end{acronym}

\usepackage{subfig}
\usepackage{xcolor}
\usepackage{float}
\usepackage{svg}
\usepackage{pdfpages}
\usepackage{comment}

\usepackage{graphicx,wrapfig,lipsum}
\usepackage{rotating}
\usepackage{amsmath, amssymb, amsthm}
\usepackage{stmaryrd}
\usepackage{tikz}
\usepackage{verbatim}
\usepackage{url}
\usepackage[utf8]{inputenc}
\usepackage[english]{babel}

\newtheorem{theorem}{Theorem}[]
\newtheorem{corollary}{Corollary}[]

\newtheorem{remark}{Remark}[]

\newtheorem{lemma}[]{Lemma}

\usepackage{multicol}

\usepackage{scalerel}
\usepackage{multicol}
\usepackage{setspace}
\newtheorem{definition}{Definition}
\usepackage{booktabs}
\usepackage{bbm}
\usepackage[normalem]{ulem}
\usepackage{color}
\usepackage{dsfont}
\usepackage{bm}
\usepackage{setspace}
\usetikzlibrary{automata}
\usetikzlibrary{arrows}
\usetikzlibrary{shapes.geometric, arrows}
\usepackage[]{footmisc}
\usetikzlibrary{positioning}
\usetikzlibrary{calc}
\usetikzlibrary{arrows}
\allowdisplaybreaks

\usepackage{lipsum}

\usepackage[margin = 0.3cm]{caption}
\usepackage{mathtools}
\usepackage{csquotes}

\makeatletter
\newcommand{\vast}{\bBigg@{3}}
\newcommand{\Vast}{\bBigg@{4}}
\makeatother

\begin{document}

\title{Fine Grained Analysis and Optimization of Large Scale Automotive Radar Networks}

\author{
\IEEEauthorblockN{Mohammad~Taha~Shah}, {\it Graduate Student Member, IEEE}, \IEEEauthorblockN{Gourab~Ghatak}, {\it Member, IEEE}, and \IEEEauthorblockN{Shobha~Sundar~Ram}, {\it Senior Member, IEEE}
\thanks{M.T. Shah is with the Bharti School of Telecommunication Technology and Management, IIT Delhi, New Delhi, India, 110016; Email: tahashah@dbst.iitd.ac.in. G. Ghatak is with the Department of Electrical Engineering, IIT Delhi, New Delhi, India, 110016; Email: gghatak@ee.iitd.ac.in. S. S. Ram is with the Department of Electronics and Communication Engineering, IIIT Delhi, New Delhi, India, 110020; Email: shobha@iiitd.ac.in}
\vspace{-1cm}}

\maketitle

\begin{abstract}
Advanced driver assistance systems (ADAS) enabled by automotive radars have significantly enhanced vehicle safety and driver experience. However, the extensive use of radars in dense road conditions introduces mutual interference, which degrades detection accuracy and reliability. Traditional interference models are limited to simple highway scenarios and cannot characterize the performance of automotive radars in dense urban environments. In our prior work, we employed stochastic geometry (SG) to develop two automotive radar network models: the Poisson line Cox process (PLCP) for dense city centers and smaller urban zones and the binomial line Cox process (BLCP) to encompass both urban cores and suburban areas. In this work, we introduce the meta-distribution (MD) framework upon these two models to distinguish the sources of variability in radar detection metrics. Additionally, we optimize the radar beamwidth and transmission probability to maximize the number of successful detections of a radar node in the network. Further, we employ a computationally efficient Chebyshev-Markov (CM) bound method for reconstructing MDs, achieving higher accuracy than the conventional Gil-Pelaez theorem. Using the framework, we analyze the specific impacts of beamwidth, detection range, and interference on radar detection performance and offer practical insights for developing adaptive radar systems tailored to diverse traffic and environmental conditions.
\end{abstract}


\section{Introduction}
\label{sec:intro}
Advanced driver assistance systems (ADAS) use automotive radars to improve safety through features such as adaptive cruise control, obstacle detection, and blind spot monitoring~\cite{bilik2019rise,lu2014connected}. The proliferation of these radars, however, introduces mutual interference between vehicles deteriorating the overall target detection accuracy~\cite{goppelt2010automotive, alland2019interference}. Traditionally, automotive radar interference has been studied using simple distributions of vehicles on highways using ray-tracing~\cite{schipper2015simulative}. Similarly, recent works have employed \ac{SG} to model automotive radar networks on highways, analyzing detection performance in terms of mean \ac{SIR}~\cite{al2017stochastic}. However, these models do not accurately emulate complex urban environments with diverse street geometries and vehicular conditions. Accurate spatial distribution models of automotive radar vehicles are essential for optimizing radar network performance and effectively mitigating mutual interference. 

In our previous work, we considered two \ac{SG}-based frameworks for modeling complex urban environments: the doubly stochastic homogeneous \ac{PLCP} and the non-homogeneous \ac{BLCP}~\cite{shah2024modeling}. 
Here, the line processes model the random distribution of streets within an area, while point processes model the random distribution of vehicles on a particular street. 
The \ac{PLCP} has uniform street and vehicular densities, suitable for small local regions within a city, while the \ac{BLCP} accounts for the heterogeneity of vehicular densities in areas with hierarchical street structures (e.g., highways, local streets, intersections), making it more appropriate for citywide scenarios.
These models enable researchers to derive analytical expressions for the average detection probability of all the nodes in the network based on the \ac{SIR} distribution~\cite{shah2024modeling}. 

For a deeper understanding, we must examine the detection performance of an individual radar rather than the average performance of all radars in the network. In particular, the simplistic \ac{SIR} distribution combines all the random factors affecting a network, such as the density of interfering radars, channel fading conditions, etc., making it infeasible to isolate the impact of individual random elements on the overall radar detection performance. 
To address these limitations, the authors introduced the concept of \ac{MD} in~\cite{haenggi2021meta,haenggi2021meta2}, which separates the different sources of randomness. A key challenge is the fact that closed-form expressions for the \ac{MD} can only be derived for some simple scenarios, such as when considering signal power or \ac{SIR}~\cite{haenggi2015meta}, using the Gil-Pelaez theorem~\cite{gil1951note}. However, for more complex models, it is generally infeasible to directly obtain analytic expressions. 
Thanks to the recent research presented in~\cite{wang2023fast}, several methods are now available for reconstructing MDs from the moments of \ac{CSP}.

{\bf Related Works:} Modeling and characterization of radar interference has a rich literature, e.g., see~\cite{goppelt2011analytical, schipper2015simulative, al2017stochastic, munari2018stochastic, chu2020interference}. In~\cite{brooker2007mutual}, the authors demonstrated significant degradation in detection due to the interference, while authors in~\cite{goppelt2011analytical} discussed the impact of ghost targets and reduced sensitivity from \ac{FMCW} radars on radar performance. Schipper {\it et al}~\cite{schipper2015simulative} analyzed radar interference using traffic flow patterns to model vehicle distribution along roadways. The authors in~\cite{xu2017interference} defined the spectral density distribution of orthogonal noise waveforms using an optimized Kaiser function and phase retrieval technique to reduce mutual interference between automotive radars. 

More recent efforts exploited \ac{SG} techniques to characterize radar interference across varying radar distributions and propose mitigation strategies. 
For instance,~\cite{ghatak2022radar, al2017stochastic} utilized a \ac{PPP} to model vehicular one-dimensional radar distribution on highways to evaluate mean \ac{SIR}. Similarly,~\cite{munari2018stochastic} employed the strongest interferer approximation method to analyze radar detection range and false alarm rates. In~\cite{fang2020stochastic}, radar detection probability is derived for targets with fluctuating radar cross-sections (RCS) modeled with Swerling-I and Chi-square frameworks. Most of these analyses focus on two-lane traffic scenarios; however, multi-lane interference was addressed in~\cite{chu2020interference} through a marked point process model. 
Using a \ac{PPP} model, the study in~\cite{huang2019v2x} evaluated interference scenarios and proposed a centralized framework leveraging \ac{V2X} communication to allocate spectrum resources and minimize radar interference. Likewise, authors in~\cite{zhang2020vanet} introduce a vehicular ad hoc networking (VANET)-assisted scheme that employs a \ac{TDMA}-based \ac{MAC} protocol to coordinate radar spectrum access among vehicles, enhancing interference mitigation. A distributed networking protocol for wireless radar control and interference mitigation was proposed in~\cite{aydogdu2019radchat}. 
In~\cite{wang2023performance}, the authors introduced multiple spectrum-sharing strategies to mitigate interference in a vehicular radar network modeled as an \ac{MHCP} model in two-lane and multi-lane scenarios. In~\cite{wang2023performance_2} authors proposed a \ac{TDMA} scheme for coordinated interference mitigation, evaluated in terms of delay, interference-free radar capacity, interference probability,
and control overhead.

The traditional metrics in \ac{SG} are limited to providing an expected view of the system performance. In this regard, the \ac{MD} framework has been widely utilized to study fine-grained aspects of wireless communication networks,~\cite{salehi2017analysis, wang2017sir, elsawy2017meta, saha2020meta, shi2021meta, wang2018sir, sun2023fine, qin2023downlink, qin2023uplink} ranging from cellular networks, low earth orbit (LEO) satellite networks and UAV assisted networks. 
The authors in~\cite{jeyaraj2021transdimensional} proposed a trans-dimensional \ac{PPP} (TPPP) model to analyze complex vehicular communication networks while accounting for street geometry and fading conditions. 
Authors in~\cite{feng2020separability} investigated the \ac{SIR} \ac{MD} to understand the link performance in interference-limited wireless networks. They demonstrate that in Poisson networks, with independent fading and path loss, the \ac{SIR} \ac{MD} can be expressed as a product of the \ac{SIR} threshold ($\beta$) and a function of reliability ($t$) within a specific region, i.e., $(\beta, t)$ defined by the fading statistics. Despite its widespread use in analyzing cellular and ad-hoc networks, the application of \ac{MD} in studying automotive radar performance is limited, especially in complex urban geometries. 

\begin{figure}[t]
\centering
\includegraphics[trim={0cm 1cm 0cm 0cm},clip,width=0.5\textwidth]{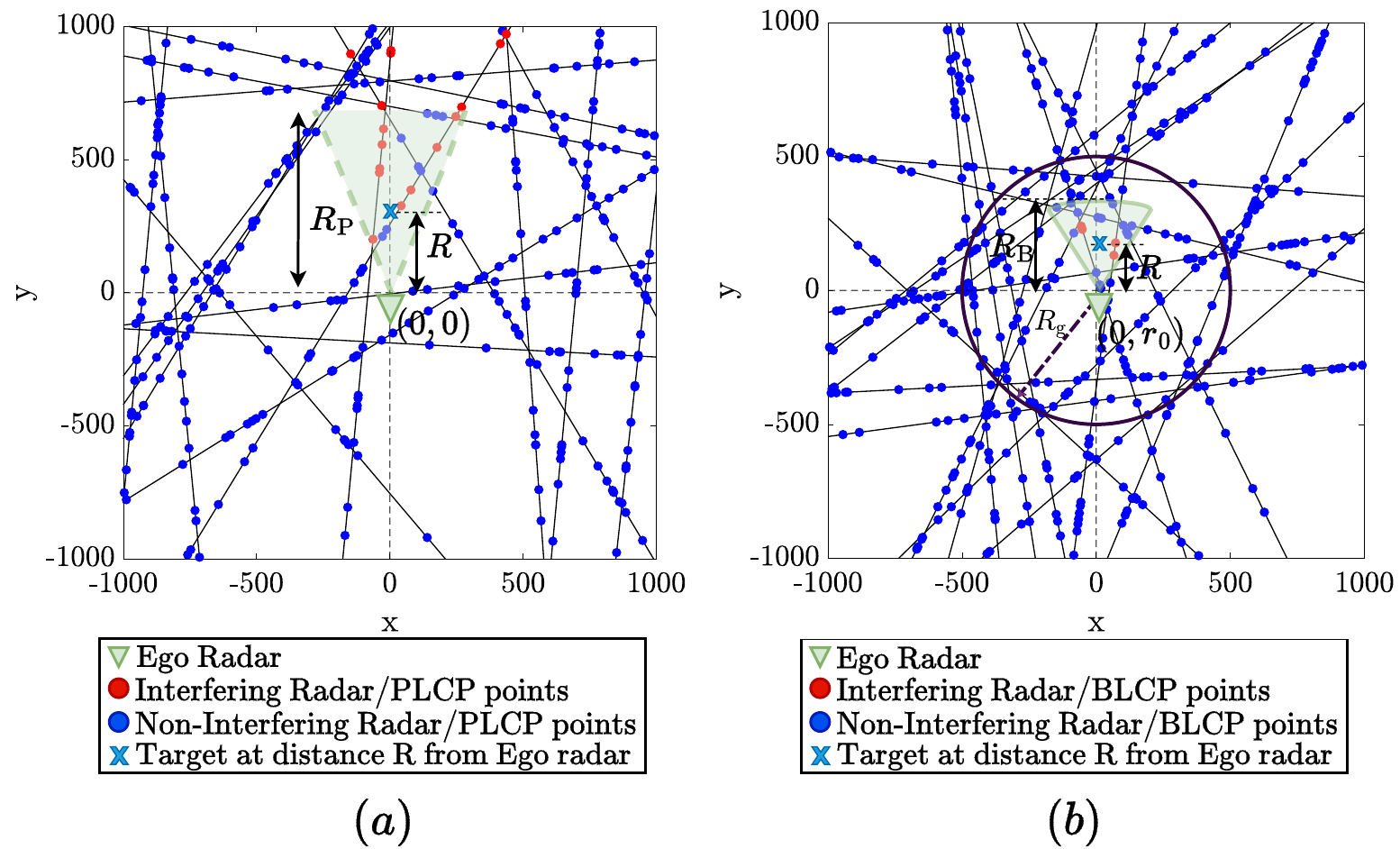}
\caption{(a) A realization of \ac{PLCP} having $\lambda_{\rm L} = 0.005 \,\rm{m}^{-2}$ and $\lambda = 0.005 \,\rm{m}^{-1}$, with ego radar present at origin, and (b) A realization of \ac{BLCP} having $n_{\rm B} = 50$ and $\lambda = 0.005 \,\rm{m}^{-1}$, with ego radar present at $(0,r_0)$.}
\label{fig:fig_1} 
\vspace{-0.2cm}
\end{figure}

\vspace{-0.2cm}
\subsection{Contributions, Organization, and Notation}
In this study, we characterize two SG models tailored for distinct urban contexts for automotive radars: (i) \ac{PLCP} model for densely populated city centers and smaller urban regions, and (ii) the \ac{BLCP} model, which encompasses both central and suburban areas to offer a comprehensive view of the city. For each model, we derive analytical frameworks to evaluate the number of successful detections by an ego radar, the corresponding detection probability, and optimal system parameters to maximize the detection performance.
The major contributions of this work are summarized as follows:
\begin{itemize}
    \item We first derive the expected number of successful detections by an ego radar within both the \ac{PLCP} and the \ac{BLCP} frameworks. This is achieved by utilizing the average lengths of \ac{PLP} and \ac{BLP} lines present within the bounded radar sector, providing a robust foundation for understanding radar performance. Subsequently, we propose an optimization framework to calculate the optimal beamwidth.
    \item We derive an analytical expression of the \ac{MD} of \ac{SF} and \ac{SIR}. Subsequently, we reconstruct the \ac{MD} using moments leveraging techniques such as the Chebyshev-Markov (CM) bound method. This approach offers insights into the individual detection performance of vehicles, enhancing our understanding of radar effectiveness in real-world scenarios.
    \item Utilizing the first negative moment of \ac{SF} \ac{MD}, we derive metrics for the first successful detection, leveraging which we optimize transmission probability.
    This optimization benefits cognitive radars, enabling adaptive adjustments to transmission strategies based on real-time environmental conditions and interference levels.
\end{itemize}

In section~\ref{sec:sys}, we first introduce the network geometry and define the channel model along with the channel access scheme. In Section~\ref{sec:AvgTar}, we derive the average number of potential targets that lie within the radar sector for both the \ac{PLCP} and \ac{BLCP} models. In section~\ref{sec:meta}, we derive the \ac{MD} of \ac{SF} 
and define the \ac{MD} reconstruction process. In section~\ref{sec:Results}, we plot all the numerical results highlighting the optimization strategies and the insights a network operator can draw from them. Finally, the paper concludes in section~\ref{sec:con}.
To differentiate the notations between the two Cox process models, we use subscript of `$k = {\rm P}$' for \ac{PLCP} and `$k = {\rm B}$' for \ac{BLCP} models. For example, line processes are identified using calligraphic letters such as $\mathcal{P}$, thus $\mathcal{P}_{\rm P}$ denotes \ac{PLP} and $\mathcal{P}_{\rm B}$ denotes \ac{BLP}. Likewise, point processes are identified using the symbol $\Phi_{\rm P}$ for \ac{PLCP}, and $\Phi_{\rm B}$ for \ac{BLCP}. A \ac{BLP} usually consists of $n_{\rm B}$ lines. 

\vspace{-0.2cm}
\section{System Model and Network Geometry}
\label{sec:sys}
Urban street networks show variability in street density and orientation, posing challenges in modeling, analyzing, and optimizing automotive radar system performance. In our previous work~\cite{shah2024modeling} we highlighted the non-homogeneity of real-world road networks 
Table II of~\cite{shah2024modeling} illustrates a stochastic distribution of streets characterized by both density and orientation. To address this, we model the street network as a stochastic line process. Specifically, we use a \ac{PLP} to represent local city regions with relatively homogeneous street distributions
and a \ac{BLP} to capture the variations in street density between urban centers and suburbs.

\subsubsection{\ac{PLCP}} In the \ac{PLP} model, streets are represented as a set of random lines, $\mathcal{P}_{\rm P} = \{L_1, L_2, \dots\}$, in the two-dimensional $xy$ Euclidean plane. One instance of this distribution is shown in Fig.~\ref{fig:fig_1}(a). Each line $L_i$ is defined by its distance $r_i$ from the origin and the angle $\theta_i$ between its normal and the $x$-axis. The parameters $(\theta_i, r_i)$ correspond to a point $q_i$ in the domain space defined as $\mathcal{D}_{\rm P} \equiv [0, \pi) \times (-\infty, \infty)$, creating a one-to-one correspondence between lines in $\mathcal{P}_{\rm P}$ and points in $\mathcal{D}_{\rm P}$. The number of points in any subset, $S \subset \mathcal{D}_{\rm P}$, follows a Poisson distribution with parameter $\lambda_{\rm L} |S|$, where $|S|$ is the Lebesgue measure of $S$, and $\lambda_{\rm L}$ represents urban street density.
Without loss of generality, we position the ego radar at the origin, represented by the green triangle in the figure, with its antenna beam also shown in green. From Slivnyak’s theorem, conditioning on the location of a point in a \ac{PPP} is the same as an addition of a point at the origin to the \ac{PPP} in the $\mathcal{D}_{\rm P}$~\cite{haenggi2012stochastic}. This results in a \ac{PLP}, $\mathcal{P}_{{\rm P}_0} = \mathcal{P}_{\rm P} \cup L_0$, where $L_0$ corresponds to $(\theta_0, r_0) = (0,0)$ and denotes the street/line containing the typical ego radar located at the origin.

On each line $L_i$, we model vehicles as an independent 1D \ac{PPP}, $\Phi_{L_i}$, with intensity $\lambda$, representing vehicular density. These are represented by blue dots in the figure. This leads to the construction of \ac{PLCP}, $\Phi_{\rm P} = \bigcup_{L_i \in \mathcal{P}_{\rm P}} \Phi_{L_i}$ which is a doubly stochastic Poisson process. Following the Palm distribution of the \ac{PLCP}, the point process is given by $\Phi_{{\rm P}_0} = \Phi_{\rm P} \cup \Phi_{L_0}$, where the line $L_0$ accounts for the traffic moving in the opposite direction (towards the ego vehicle) on the street where the ego vehicle is situated. 
Interference occurs when the ego radar and another vehicle's radar are within each other’s beam sectors (to be defined soon). The interfering radars are indicated by red dots in the figure. 

\subsubsection{\ac{BLCP}}
The BLP model considers the network of streets as a finite set of lines, $\mathcal{P}_{\rm B} = \{L_1, L_2, \dots, L_{n_{\rm B}}\}$, in the 2D Euclidean plane. One such instance is shown in Fig.~\ref{fig:fig_1}(b). Each line $L_i$ corresponds to a point, $(\theta_i, r_i) \in \mathcal{D}_{\rm B} \equiv [0, \pi) \times [-R_g, R_g]$, where $\mathcal{D}_{\rm B}$ is the domain space of \ac{BLP}. Note that, unlike the PLP case, the generating points are restricted to a disk of radius $R_g$ centered at the origin. The vehicle positions along each line $L_i$ are modeled as independent 1D PPPs, $\Phi_{L_i}$, with intensity $\lambda$. The complete distribution forms \ac{BLCP} defined as $\Phi_{\rm B} = \bigcup_{i=1}^{n_{\rm B}} \Phi_{L_i}$. The overall process, accounting for vehicles in the opposite direction of the ego vehicle's street, is represented as $\Phi_{{\rm B}_0} = \Phi_{\rm B} \cup L_0$. Unlike the homogeneous \ac{PLCP}, the inhomogeneous \ac{BLCP} cannot be characterized from a single typical point but depends on the distance from the origin. Thus, without loss of generality, we consider the ego radar to be located at $(0, r_0)$, where $r_0$ is its distance from the origin. The maximum unambiguous range of the radar is $R_{\rm B}$ (to be defined soon). As in the \ac{PLCP} scenario, interference occurs if the ego and interfering radars are within each other’s maximum unambiguous range and beam width.

Figures.~\ref{fig:fig_1}(a) and (b) highlight the stark contrast between the spatial structures of the \ac{PLCP} and \ac{BLCP}. In a \ac{PLCP}-modeled network, the ego radar experiences uniform interference characteristics regardless of its location. However, in a \ac{BLCP}-modeled network, the interference characteristics vary depending on the ego radar's position in the Euclidean plane. Specifically, in the \ac{BLCP}, an ego radar located near the city center (i.e., within the generating circle) will encounter different statistical features, such as street and intersection density, compared to those observed outside the city center.

\subsection{Interfering Set}
We assume that the radars are mounted on the front and rear of a vehicle and $\Omega$ is the half-power beamwidth. The orientation of vehicles on $L_i$ depends on the generating angle, $\theta_i$, and the direction of the vehicular movement. Thus, the boresight direction of any radar on $L_i$ is given by the two unit vectors: $\mathbf{a}$ and $-\mathbf{a}$, where $\mathbf{a} = (-\sin\theta_i, \cos\theta_i)^T$. 
Assume that any radar is located at $(x,y)$ such that $(x,y) \in \Phi_{k_0}$. Any interfering radar represented as a point in the Euclidean plane $(p,q) \in \Phi_{k_0}$, lies in the antenna's beam sector of radar located at $(x,y)$ if the angle made by the displacement vector between $(p,q) - (x, y)$ and the boresight direction exceeds $\cos \Omega$. Therefore, the radar sector is uniquely characterized by $\mathcal{R}^{+}_{(x,y), k}$ and $\mathcal{R}^{-}_{(x,y), k}$ as a function of $(x,y)$ and $\Omega$. Based on this, we define the \textit{closed} interior region of the radar sector for any Cox point as
\begin{definition}
\label{def:def_1}
The radar sector of a radar located at $(x,y) \in \Phi_{k_0}$ such that $x\cos\theta_i+y\sin\theta_i=r_i$ and $(\theta_i,r_i) \in \mathcal{D}_k$ for boresight direction $\mathbf{a}$ is
\begin{align*}
    \mathcal{R}^{+}_{(x,y), k} &= \Bigg\{(p,q) \in \mathbb{R}^2 \colon \frac{\big((p,q) - (x, y)\big) \cdot \mathbf{a}}{||(p,q) - (x, y)||} > \cos \Omega, \\
    &\hspace*{3.5cm} ||(p,q) - (x, y)|| \leq R_k \Bigg\}.
\end{align*}
\end{definition}
In the case of \ac{PLCP}, the radar sector of ego radar is characterized by $(\theta_i,r_i) = (0,0)$, and $(x,y) = (0,0)$ since the ego radar is situated at the origin with its street oriented towards the $y$-axis. For \ac{BLCP} the ego radar has $(x,y) = (0,r_0)$. The radar sectors formed by the automotive radar is \textit{bounded}. For both \ac{PLCP} and \ac{BLCP} analysis, we assume the radar has a maximum operational range $R_k$ where $k = {\rm P}$, or $k = {\rm B}$, meaning it cannot detect targets beyond this range. In real-world conditions, radar detection is also limited by blockages like buildings. Thus $R_k$ represents the line-of-sight (LOS) range, and any radar outside this range will not interfere with the ego radar’s detection performance. 
The ego radar at any position only experiences interference from the Cox points that lie within its radar sector, and the ego radar is in their radar sector. If the ego radar has $\mathbf{a}$ as the boresight direction, then only the radars with boresight direction $-\mathbf{a}$ can contribute to interference. We generalize the set of interfering points as follows:
\begin{definition}
\label{def:def_2}
For any $\Phi_{k_0}$, the set of Cox points $\Phi_{k}^{\rm I}$ that cause interference at the ego radar is
\begin{align*}
    &\Phi_k^{\rm I} = \bigcup_{(p,q) \in \Phi_{k_0}} \Bigg\{(p,q) \colon\\
    &\hspace*{1.5cm} \left(\mathbf{1} \left((0,r) \in \mathcal{R}^{+}_{(p,q),k}\right) + \mathbf{1} \left((0,r) \in \mathcal{R}^{-}_{(p,q),k}\right)\right) \\
    &\hspace*{0.5cm} \cdot \left(\mathbf{1} \left((p,q) \in \mathcal{R}^{+}_{(0,r),k}\right) + \mathbf{1} \left((p,q) \in \mathcal{R}^{-}_{(0,r),k}\right)\right) = 1 \Bigg\}.
\end{align*}
where $\mathbf{1}(\cdot)$ is the indicator function, $\mathcal{R}^{+}_{(p,q),k}$ is  closed interior region of any radar present at $(p,q)$, and $(0,r)$ is the location of ego radar, and $k \in \{{\rm P},{\rm B}\}$. In the case of \ac{PLCP}, the ego radar is at the origin, thus $r=0$, and in \ac{BLCP}, $r=r_0$.
\end{definition}
The above definition outlines the set of interfering automotive radars $\Phi_k^{\rm I}$, having coordinates $(p,q)$ that cause interference at the ego radar. If ego radar is within the radar sector of automotive radar present at $(p,q)$ i.e., $\mathbf{1} \left((0,r) \in \mathcal{R}^{+}_{(p,q),{\rm P}}\right) + \mathbf{1} \left((0,r) \in \mathcal{R}^{-}_{(p,q),{\rm P}}\right)$, and the automotive radar located at $(p,q)$ is in the interior region of ego radar i.e., $\mathbf{1} \left((p,q) \in \mathcal{R}^{+}_{(0,r),{\rm P}}\right) + \mathbf{1} \left((p,q) \in \mathcal{R}^{-}_{(0,r),{\rm P}}\right)$, only then radar at $(p,q)$ will be an element of $\Phi_k^{\rm I}$.

\subsection{Channel Access, \ac{SIR}, and \ac{SF}}
In our previous work~\cite{shah2024modeling}, we assumed that all the neighboring radars transmit continuously, which creates a very high probability of disturbance due to mutual interference. In civilian radar networks of high density, one consideration for tackling mutual interference is to enable vehicular radars to cognitively turn on transmissions based on ambient channel conditions. In this work, we study this concept by considering a simple ALOHA access protocol~\cite{abramson1970aloha}.  All the vehicles, including the ego radar, attempt to detect a target with probability $p$ in this scheme. We refer to $p$ as the transmission probability, as the vehicle needs to transmit the radar signal to be able to perform detection. Thus, each interference term is weighted using an indicator function $\mathbf{1} (\cdot)$ to denote whether the vehicle is transmitting or not. Let the set of locations of the interfering nodes be denoted by $\mathcal{C} \subset \Phi_{k}^{\rm I}$. 
The ego radar receives the reflected signal power from the target vehicle with strength
    $S = \gamma\sigma_{\mathbf{c}} P R^{-2\alpha}$,
where, $\gamma = \frac{G_{\rm t}}{(4\pi)^2}A_{\rm e}$,  $A_{\rm e}$ is the effective area of the receiving antenna aperture, $G_t$ is the gain of the transmitting antenna, $P$ is the transmit power, $\alpha$ is the path-loss exponent, and $\sigma_{\mathbf{c}}$ is fluctuating radar cross-section of target at distance $R$.
The interference at the ego radar from either \ac{PLCP} or \ac{BLCP} caused by a single interfering radar situated at a distance $w_k$ from the ego radar is expressed as
    $\mathbf{I}_k = P \gamma h_{\mathbf{w}_k} ||\mathbf{w}_k||^{-\alpha},$
where $h_{\mathbf{w}_k}$ is the fading power. Accordingly, the \ac{SIR} at the ego radar is
\begin{align}
    \xi_k = \frac{\gamma \sigma_{\mathbf{c}} P R^{-2\alpha}}{\sum_{\mathbf{w}_k \in \Phi_k}4\pi \gamma P h_{\mathbf{w}_k} ||\mathbf{w}_k||^{-\alpha}\textbf{1}({\bf x} \in \mathcal{C})}.
\end{align}
In the \ac{MD} analysis, we utilize \ac{SF} rather than \ac{SIR} as the main parameter. The justification for this choice is clarified in Section~\ref{sec:meta}.
\ac{SF} is defined as the ratio of signal power to total received power.
i.e., $\mathrm{SF} = \frac{S}{S+\mathbf{I}_{\rm P}}$, or $\mathrm{SF} = \frac{\rm SIR}{{\rm SIR} + 1}$. In this scheme, the \ac{SF} is defined as,
\begin{align}
    \hspace*{-0.2cm}\mathrm{SF}_k \!\!=\!\! \frac{\gamma \sigma_{\mathbf{c}} P R^{-2\alpha}}{\gamma \sigma_{\mathbf{c}} P R^{-2\alpha} \!+\! \sum_{\mathbf{w}_k \in \Phi_k}\! 4\pi \gamma P h_{\mathbf{w}_k}\! ||\mathbf{w}_k||^{-\alpha}\textbf{1}({\bf x} \in \mathcal{C})},
    \label{eq:eq_sf}
\end{align}
where $\mathcal{C}$ is the set of locations of active interfering vehicles.
\begin{figure}[t]
\centering
\includegraphics[trim={0cm 1cm 3.5cm 0.5cm},clip,width = 0.35\textwidth]{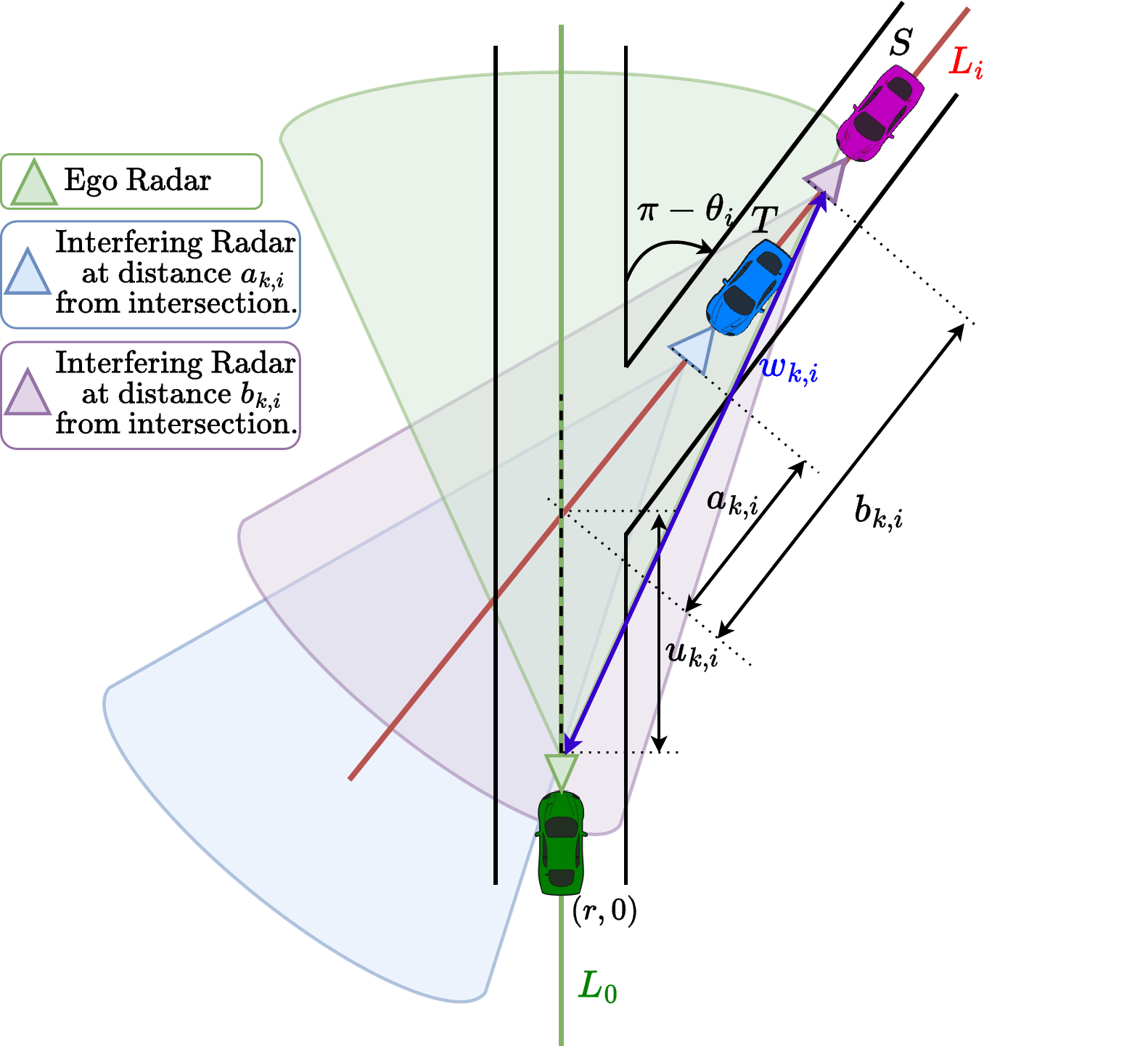}
\caption{Illustration of a scenario where two radars are present at the edge point of the line $L_i$ inducing interference.}
\label{fig:fig_3}
\end{figure}

\subsection{Interfering Distance}
We observed from Fig~\ref{fig:fig_1}. (a) and (b) that all the automotive radars emulated as Cox points do not cause interference (indicated by blue dots). Instead, interference is caused at the ego radar only when both the ego radar and the interfering radar are simultaneously within each other's radar sectors (indicated by red dots). The impact of the interfering radars located on the ego radar's street and the remaining streets on the detection performance of the ego radar is determined by the interfering distance $v_{{\rm B},i}$, \emph{the section of $L_i$ wherein if a radar is present, it will contribute to the interference experienced by the ego radar.} In order to derive $v_{{\rm B},i}$, we note that the distance from the ego radar to the intersection of the $L_i\textsuperscript{th}$ line is $d_i = u_i - r_0 = \frac{r_i}{\sin\theta_i} - r_0$, where the angle of intersection is equal to the generating angle $\theta_i$ of the intersecting line. The distance between the interfering radar on $L_i$ from the intersection point is $v_{k,i}$, referred to as the interfering distance. Figure~\ref{fig:fig_3} illustrates an ego radar with a green radar sector on $L_1$ and radars mounted on other vehicles, \textit{S} and \textit{T}, on $L_i$, with red and purple beams. Specifically, \textit{S} and \textit{T} are at a distance $v_{k,i} = b_{k,i}$ and $v_{k,i} = a_{k,i}$ from the intersection point, respectively. These two distances ($a_{k,i}$ and $b_{k,i}$) represent the bounds of $v_{k,i}$ within which another radar will interfere with the ego radar. From Fig.~\ref{fig:fig_3}, any other vehicle present behind the vehicle \textit{S} will not cause interference. Likewise, any vehicle present after vehicle \textit{T} also will not cause interference. This is because the vehicles present behind \textit{S} and after \textit{T} will not mutually interfere with ego radar. The values of $a_{k,i}$ and $b_{k,i}$ for the different scenarios are presented in Lemma 1 and Theorem 1 of~\cite{shah2024modeling} for \ac{PLCP} and \ac{BLCP} respectively. Correspondingly, the distance from the ego radar to the interfering radar is then given as
\begin{align}
    w_{k, i} &=
    \begin{cases}
        v_{k, i} & i=0 \nonumber \\
        \sqrt{(d_i + v_{k, i}|\cos\theta_i|)^2 + (v_{k, i}\sin\theta_i)^2}\; &  \mathrm{otherwise}. \nonumber
    \end{cases}
\end{align}

\begin{figure}[t]
    \centering
    \includegraphics[trim={3cm 1.5cm 0.5cm 1.5cm},clip,width=0.35\textwidth]{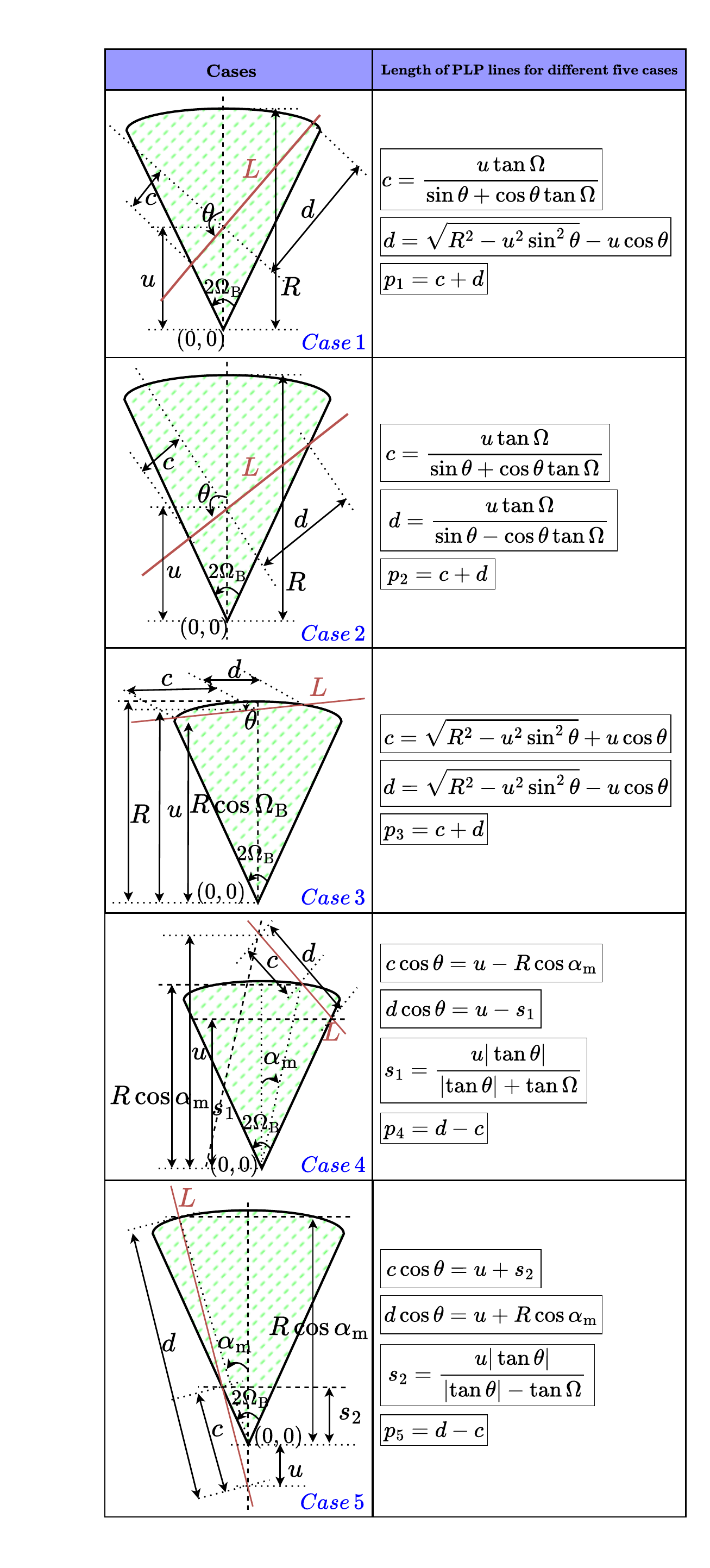}
    \caption{The figure illustrates different cases of a \ac{BLP} line intersecting line $L_0$. Case $1$ to $5$ correspond to 5 different cases to find $l$, i.e., from $l = p_1$ to $l = p_5$.}
    \label{fig:lemma_1_img}
\end{figure}

\section{Average number of potential targets
}
\label{sec:AvgTar}
In this section, we determine the average number of potential targets that are located within the radar antenna's main lobe. This involves characterizing the average number of Cox
points $n_k(R)$ in the beam sector of the ego radar with the distance to target as $R$. We identify $\mathcal{N}_{(0,r)}^{+}\!(R)$ as a collection of points $(p, q)$, lying within the target range of ego radar,
$\mathcal{N}_{(0,r)}^{+} (R) = \left\{(p,q) \in \mathbb{R}^2 \colon \frac{q}{\sqrt{p^2 + \left(q-r\right)^2}} > \cos \Omega, \sqrt{p^2 + \left(q-r\right)^2} \leq R \right\}$,
where $r=0$ in case of \ac{PLCP} and $r = r_0$ in case of \ac{BLCP}. Next, we note that $n_k(R)$  depends on the average length of \ac{PLP} or \ac{BLP} lines inside $\mathcal{N}_{(0,r)}^{+}\!(R)$. 
\begin{definition}
\label{def:def_3}
The average number of Cox points inside $\mathcal{N}_{(0,r)}^{+}\!(R)$ is
\begin{align}
    n_k(R) = \lambda l_k(R) = \lambda \mathbb{E}_{\mathcal{P}_{k_0}} \bigg[\left|L \cap \mathcal{N}_{(0,r)}^{+}(R)\right|_1\bigg],
    \label{eq:n_R}
\end{align}
where $|\cdot|_1$ is the Lebesgue measure in one dimension and $L$ is a line of the \ac{PLP} or \ac{BLP}.
\end{definition}

\subsection{Average number of \ac{PLCP} points in $\mathcal{N}_{(0,0)}^{+} (R)$}
First, we determine the average length of a single line within $\mathcal{N}_{(0,0)}^{+}\!(R)$, in the following Lemma~\ref{le:lemma_2}.
\begin{lemma}
\label{le:lemma_2}
The length of a line parameterized by $(\theta,r)$ present inside the radar sector $\mathcal{N}_{(0,0)}^{+}\!(R)$ of half beamwidth $\Omega$ and distance to target $R$ is,
\begin{align}
    l &= \nonumber\\
    &\hspace*{0cm}
    \begin{cases}
        p_0 \!=\! R_{\rm P}; \hspace*{0.5cm}\mathrm{for}\; \left(\theta , r\right) = \left(0, 0\right)\\
        p_1 \!=\! \frac{u\tan\Omega}{|\sin\theta| + |\cos\theta|\tan\Omega} + \sqrt{R^2 -u^2\sin^2\theta} - u|\cos\theta|; \\
        \hspace*{0.75cm}\mathrm{for}\; 0 \leq u \leq R, \;\mathrm{and}\; \theta \in  \big\{[0,\alpha_{\rm n}] \cup [\pi-\alpha_{\rm n}, \pi]\big\}\\
        p_2 \!=\! \frac{2 u \sin\theta \tan\Omega}{\sin^2\theta - \cos^2\theta\tan^2\Omega}; 
        \hspace*{0.5cm}\mathrm{for}\; 0 \leq u \leq R\cos\Omega, \\
        \hspace*{3.5cm}\mathrm{and}\; \theta \in [\alpha_{\rm n}, \pi-\alpha_{\rm n}] \\
        p_3 \!=\! 2\sqrt{R^2 -u^2\sin^2\theta}; \hspace*{0.6cm}\mathrm{for}\;  R\cos\Omega \leq u \leq R \\
        \hspace*{3.5cm}\mathrm{and}\; \theta \in [\alpha_{\rm n}, \pi-\alpha_{\rm n}] \\
        p_4 \!=\! \left(R \cos\alpha_{\rm m} - \frac{u|\tan\theta|}{|\tan\theta| + \tan\Omega}\right) |\sec\theta|; \\
        \hspace*{0.2cm}\mathrm{for}\;  R < u \leq R(\cos\Omega + |\cot\theta|\sin\Omega), \;\mathrm{and}\; \theta \in [0,\pi]\\
        p_5 \!=\! \left(R \cos\alpha_{\rm m} - \frac{u|\tan\theta|}{|\tan\theta| - \tan\Omega}\right) |\sec\theta|; \\
        \hspace*{0.1cm}\mathrm{for}\;  R(\cos\Omega - |\cot\theta|\sin\Omega) \leq u < 0 , \;\mathrm{and}\; \theta \in (\pi,2\pi]\\
        0; \hspace*{3.2cm}\mathrm{otherwise}
    \end{cases}
    \label{eq:lengthofline}
\end{align}
where $\alpha_{\rm n} = \arctan{\left(\frac{R\sin\Omega}{|R\cos\Omega - u|}\right)}$, and $\cos\alpha_{\rm m} = \frac{u}{R}\sin^2\theta + \sqrt{\frac{u^2}{R^2}(\sin^4\theta - \sin^2\theta) + \cos^2\theta}$.
\end{lemma}

\begin{figure*}[b!] 
    \begin{align}
    \hline
        l_{\rm B}(R) =
        \begin{cases}
        \frac{n_{\rm B}}{2 R_{\rm g}} \Omega R^2; \mathrm{for}\; r_0 \in [-R_{\rm g}, R_{\rm g}-R] \\
        \frac{n_{\rm B}}{2R_{\rm g}}\bigg[R^2 \arcsin{\left(\frac{\rm y_B}{R}\right)} + {\rm y_B}\sqrt{R^2 - {\rm y^2_B}} + R^2 \arcsin{\left(\frac{\rm y_A}{R}\right)} + {\rm y_A}\sqrt{R^2 - {\rm y^2_A}} + R_{\rm g}^2 \arcsin{\left(\frac{\rm y_A}{R_{\rm g}}\right)} + \\
        \hspace*{0cm} {\rm y_A}\sqrt{R_{\rm g}^2 - {\rm y^2_A}} - 2r_0{\rm y_A} - m{\rm y^2_A}\bigg] + \frac{2n_{\rm B}}{\pi R_{\rm g}}\int_0^{\rm y_A} \int_{\sqrt{R^2 - x^2}}^{\sqrt{R_{\rm g}^2 - x^2}+r_0} \arcsin{\left(\frac{R_{\rm g}}{\sqrt{x^2+y^2}}\right)}\, {\rm d}y {\rm d}x; \\
        \hspace{1cm} \mathrm{for}\; r_0 \in \left(R_{\rm g}-R, \sqrt{R_{\rm g}^2 - \left(R \sin \Omega \right)^2} - R\cos\Omega\right]\\
        \frac{n_{\rm B}}{2R_{\rm g}}\bigg(R_{\rm g}^2 \arcsin{\left(\frac{\rm y_C}{R_{\rm g}}\right)} + {\rm y_C}\sqrt{R_{\rm g}^2 - {\rm y^2_C}} - m{\rm y^2_C} - 2r_0{\rm y_C}\bigg) + \frac{2n_{\rm B}}{\pi R_{\rm g}}\bigg[\int_0^{\rm y_C} \int_{\sqrt{R_{\rm g}^2 - x^2}}^{\sqrt{R^2 - x^2}+r_0} \arcsin{\left(\frac{R_{\rm g}}{\sqrt{x^2+y^2}}\right)} {\rm d}y {\rm d}x + \\
        \hspace*{3cm}\int_{\rm y_C}^{\rm y_B} \int_{mx+r_0}^{\sqrt{R^2 - x^2}+r_0} \arcsin{\left(\frac{R_{\rm g}}{\sqrt{x^2+y^2}}\right)}\, {\rm d}y {\rm d}x \bigg]; \mathrm{for}\; r_0 \in \left(\sqrt{R_{\rm g}^2 - \left(R \sin \Omega \right)^2} - R\cos\Omega, R_{\rm g}\right]\\
        \frac{2n_{\rm B}}{\pi R_{\rm g}} \int_{0}^{\rm y_B} \int_{mx+r_0}^{\sqrt{R^2 - x^2}+r_0} \arcsin{\left(\frac{R_{\rm g}}{\sqrt{x^2+y^2}}\right)}\, {\rm d}y {\rm d}x; \mathrm{for}\; r_0 \in \big\{\left(R_{\rm g},\infty\right) \cup \left(-\infty, -(R_{\rm g} + R)\right)\big\}\\
        \frac{n_{\rm B}}{2R_{\rm g}}\bigg[R^2 \arcsin{\left(\frac{\rm y_B}{R}\right)} + {\rm y_B}\sqrt{R^2 - {\rm y^2_B}}  + R_{\rm g}^2 \arcsin{\left(\frac{\rm y_C}{R_{\rm g}}\right)} + {\rm y_C}\sqrt{R_{\rm g}^2 - {\rm y^2_C}} + 2r_0{\rm y_C} - m\left({\rm y^2_B} - {\rm y^2_C}\right) \bigg] - \\
        \hspace*{3cm} \frac{2n_{\rm B}}{\pi R_{\rm g}} \int_0^{\rm y_C} \int_{\sqrt{R_{\rm g}^2 - x^2}}^{mx+r_0} \arcsin{\left(\frac{R_{\rm g}}{\sqrt{x^2 +y^2}}\right)} {\rm d}y {\rm d}x; \mathrm{for}\; r_0 \in \left[-\sqrt{R_{\rm g}^2 - \left(R \sin \Omega \right)^2} - R\cos\Omega, -R_{\rm g}\right)\\
        \frac{n_{\rm B}}{2R_{\rm g}}\bigg[R^2 \arcsin{\left(\frac{\rm y_A}{R}\right)} + {\rm y_A}\sqrt{R^2 - {\rm y^2_A}} + R_{\rm g}^2 \arcsin{\left(\frac{\rm y_A}{R_{\rm g}}\right)} + {\rm y_A}\sqrt{R_{\rm g}^2 - {\rm y^2_A}} + 2r_0{\rm y_A}\bigg] - \\
        \hspace*{0cm} \frac{2n_{\rm B}}{\pi R_{\rm g}} \bigg[\int_{0}^{-\rm y_A} \int^{mx+r_0}_{-\sqrt{R_{\rm g}^2 - x^2}} \arcsin{\left(\frac{R_{\rm g}}{\sqrt{x^2+y^2}}\right)}\, {\rm d}y {\rm d}x + \\
        \hspace*{1cm} \int_{-\rm y_A}^{\rm y_B} \int^{mx+r_0}_{-\sqrt{R^2 - x^2}+r_0} \arcsin{\left(\frac{R_{\rm g}}{\sqrt{x^2+y^2}}\right)}\, {\rm d}y {\rm d}x \bigg];  \mathrm{for}\; r_0 \in \left[-(R_{\rm g} + R), -\sqrt{R_{\rm g}^2 - \left(R \sin \Omega \right)^2} - R\cos\Omega\right)\\
    \end{cases}
    \label{eq:len_blcp}
    \end{align}
\end{figure*}

\begin{proof}
In order to determine the length of a line $L$ within the radar sector, we have to carefully consider different cases of generating angle $\theta$ and the intersecting distance $u$. Depending on point of intersection and $\theta$ there are 6 different cases to find $l$ corresponding to $p_j$ where $j = \{0,1,\dots,5\}$. The value of $l = p_0$ is due to $L_0$. Let us define four events based on the value of $\theta_i$, which are as $\mathrm{A_1} =  \left\{\theta \colon 0 \leq \theta \leq \frac{\pi}{2} \right\}$, $\mathrm{A_2} =  \left\{\theta \colon \frac{\pi}{2} \le \theta \leq \pi \right\}$, $\mathrm{A_3} =  \left\{\theta \colon  \pi \le \theta \leq \frac{3\pi}{2} \right\}$, and $\mathrm{A_4} =  \left\{\theta \colon \frac{3\pi}{2} \le \theta \leq 2\pi \right\}$, corresponding to the generating present in four quadrants. As ego radar is located at the origin, the intersection occurs ahead of the ego radar if $0 \leq \theta \leq \pi$, which corresponds to cases $1$-$4$, and the intersection occurs ahead of the ego radar if $\pi \le \theta \leq 2\pi$ which corresponds to case $5$. For case $1$, $L$ intersects one of the edge lines and circular curvature of the radar sector, while as in case $2$, $L$ intersects only the edge lines of the radar sector. In case $3$ line $L$ intersects only the arc of the sector. Case $4$ and $5$ are the same as case $1$, except that in case $4$, $L$ intersects the y-axis outside the radar sector and ahead of the ego radar, whereas, in case $5$, $L$ intersects the y-axis behind the ego radar as shown in Fig.~\ref{fig:lemma_1_img}. In this proof, we will focus only on case $1$, as the remaining cases have similar procedures to derive them. Consider for event $\mathrm{A_1}$, the value $p_1$ can be found as $p_1 = c+d$, where $c = \frac{u\tan\Omega}{\sin\theta + \cos\theta\tan\Omega}$ and $d = \sqrt{R^2 -u^2\sin^2\theta} - u\cos\theta$ follows from simple trigonometric operations. To determine $p_1$ for event $\mathrm{A_2}$ as shown in Fig.~\ref{fig:lemma_1_img}, we take the modulus of sine and cosine to accommodate the first two quadrants. Like wise for case $2$ and $3$, $p_j = c+d$, while as for case $4$ and $5$, $p_j = d-c$. By looking at Fig.~\ref{fig:lemma_1_img}, we can deduce that the remaining cases would follow a similar method of finding $c$ and $d$ while considering event $\mathrm{A_1}$ and $\mathrm{A_2}$ for cases $1$-$4$, and event $\mathrm{A_3}$ and $\mathrm{A_4}$ for case $5$.
\end{proof}

Following Lemma~\ref{le:lemma_2}, we derive the average length of $\ac{PLP}$ lines inside $\mathcal{N}_{(0,0)}^{+}\!(R)$ in Theorem~\ref{th:theo1}.
\begin{theorem}
\label{th:theo1}
In a \ac{PLP}, the average length of line segments present inside the bounded radar sector $\mathcal{N}_{(0,0)}^{+}\!(R)$  is
\begin{align}
    l_{\rm P} (R) = \mathbb{E}_{\mathcal{P}_{{\rm P}_0}} \bigg[\left|L \cap \mathcal{N}_{(0,0)}^{+}(R)\right|_1\bigg] = 2\pi\lambda_{\rm L}R\bar{l} + l_0, 
\end{align}
where
\begin{align*}
    \bar{l} &= \frac{1}{2\pi R} \vast[\int_0^{R}\!\! \int_0^{\alpha_{\rm n}} p_1\, {\rm d}\theta\, {\rm d}r + \int_0^{R\cos\Omega}\!\! \int_{\alpha_{\rm n}}^{\pi-\alpha_{\rm n}} p_2\, {\rm d}\theta\, {\rm d}r \\
    & + \int_{R\cos\Omega}^R\!\! \int_{\alpha_{\rm n}}^{\pi-\alpha_{\rm n}} p_3\, {\rm d}\theta\, {\rm d}r + \int_R^{R\left(\cos\Omega + |\cot\theta|\sin\Omega \right)}\!\! \int_{0}^{\pi}\\
    &\hspace*{1.8cm}p_4\, {\rm d}\theta\, {\rm d}r + \int_{R\left(\cos\Omega - |\cot\theta|\sin\Omega \right)}^0\!\! \int_{\pi}^{2\pi} p_5\, {\rm d}\theta\, {\rm d}r\vast].
\end{align*}
\end{theorem}
\begin{remark}
\label{rem:remark_1}
Utilizing the result of Theorem~\ref{th:theo1} in definition~\ref{def:def_3}, we observe that there is a linear relation between the average number of points in a radar sector and the intensity of the \ac{PLP}. The same linear relation can be seen with the intensity of vehicles on the road. Therefore the average number of targets increase linearly with $\lambda_{\rm L}$ and $\lambda$.
\end{remark}
\subsection{Average number of BLCP points in  $\mathcal{N}_{\left(0,r_0\right)}^{+} (R)$}
The average number of vehicles inside the ego radar's beam sector depends on the location of the radar due to the non-homogeneous nature of \ac{BLCP}. 
In order to derive $n_{\rm B}(R)$, first we derive the average length of \ac{BLP} lines in $\mathcal{N}^{+}_{\left(0, r_0\right)}(R)$. Unlike the \ac{PLP}, where the analysis follows a simpler approach of first determining the length of a single line inside the radar sector and then averaging the Poisson-distributed number of lines, the \ac{BLP} utilizes the result of Theorem 2 from~\cite{shah2024binomial} where the \textit{line length density} of \ac{BLP} is derived as a function of distance from the origin. The theorem states that:
\begin{lemma}
\label{le:radial_d}
For a \ac{BLP} having $n_{\rm B}$ lines within a disk of radius $R_{\rm g}$, the line length density at a distance $r$ from origin is,
\begin{align}
    \label{eq:lemm_3}
    \rho (r) {=} \begin{cases} 
    \frac{n_{\rm B}}{2R_{\rm g}},  &\mathrm{if}\, r \leq R_{\rm g}    \\
    \frac{n_{\rm B}}{\pi R_{\rm g}} \arcsin{\left(\frac{R_{\rm g}}{r}\right)}& \mathrm{if}\, r > R_{\rm g}.
    \end{cases}
\end{align}
\end{lemma}
The above lemma shows that $\rho (r)$ first remains constant up to $R_{\rm g}$, and then decreases as $\mathcal{O}(r)$, as $r \rightarrow \infty$. In other words, the street networks become less dense as we move towards the suburbs. We integrate the line length density for a bounded radar sector based on the location of ego radar, which
gives us the average length of \ac{BLP} lines inside $\mathcal{N}^{+}_{\left(0, r_0\right)}(R)$,
\begin{theorem}
\label{th:theo2}
In a BLP having $n_{\rm B}$ lines and generated within a disk of radius $R_{\rm g}$, the average length of line segments inside $\mathcal{N}^{+}_{\left(0, r_0\right)}(R)$ is given by $l_{\rm B}(R)$ in~\eqref{eq:len_blcp},
where
\begin{align*}
    &{\rm y_A} = \frac{\sqrt{4R_{\rm g}^2r_0^2 - \left(R_{\rm g}^2+r_0^2-R^2\right)^2}}{2r_0}, \quad m = \cot\Omega \\
    &{\rm y_B} = R\sin\Omega, \quad {\rm y_C} = \frac{-m r_0 \pm \sqrt{R_{\rm g}^2(1+m^2) - r_0^2}}{1+m^2}.
\end{align*}
\end{theorem}

\begin{proof}
    See Appendix~\ref{pr:thm_2}.
\end{proof}

\section{Meta distribution Analysis}
\label{sec:meta}
Traditional studies on radar performance analysis involved the derivation of the probability of missed detection and the probability of false alarms. But in this article, we study an automotive radar network modeled using stochastic spatial processes and use the \ac{MD} of the \ac{SIR} and \ac{SF} for fine-grained performance analysis. In~\cite{shah2024modeling}, authors relied on a single,
simplified distribution of the \ac{SIR}. However, this approach failed to capture the nuances of individual detection performance due to factors like fading and locations of the interfering vehicles. The \ac{MD} of the \ac{SIR} offers a more granular perspective by explicitly accounting for these different sources of randomness. Despite these advantages, the computation of \ac{MD} often requires focusing on the moments of the underlying conditional distribution. This section explores the application of the \ac{MD} for both \ac{PLCP} and \ac{BLCP} based radar framework. 
\subsection{Meta Distribution of SF}
Compared to the \ac{MD} of \ac{SIR}, the \ac{MD} of \ac{SF} offers a key advantage: \ac{SF} is supported on the compact interval $[0, 1]$ compared to the SIR \ac{MD}'s range that extends to positive infinity. The bounded support of \ac{SIR} \ac{MD} ensures that all positive moments of the SF always exist, eliminating the need for truncation techniques when analyzing its distribution. Our approach enables a detailed analysis of an individual automotive radar's performance within the network, focusing on per-link metrics like detection probability, delay, and performance variance. From~\eqref{eq:eq_sf}, the \ac{MD} of \ac{SF} is derived as,
\begin{align}
    F_{p_{\rm SF}, k} (t_{\rm SF}) &= \mathcal{P}_{{\rm M},k} \left(\beta_{\rm SF}, t_{\rm SF}\right) = \mathbb{P} \left(p_{{\rm SF},k}(\beta_{\rm SF}) \geq t_{\rm SF}\right) \nonumber\\
    &= \mathbb{P}\big(\mathbb{P} \left(\mathrm{SF}_k > \beta_{\rm SF} | \Phi_{k_0}\right) \geq t_{\rm SF}\big)
    \label{eq:eq_md}
\end{align}
where $\beta_{\rm SF} \in [0,1]$ is the SF threshold, and $t_{\rm SF}\in [0,1]$ is the \ac{CCDF} threshold or reliability threshold. The \ac{CCDF} of the random variable, $\mathbb{P} \left(\mathrm{SF}_k > \beta_{\rm SF} | \Phi_{k_0}\right)$, is the \ac{MD} of the \ac{SF}, given by
\begin{align}
    &p_{{\rm SF},k}(\beta_{\rm SF}) = \mathbb{P} \left(\mathrm{SF}_k > \beta_{\rm SF} | \Phi_{k_0}\right) = \nonumber\\
    & \mathbb{P}\left(\sigma_{\mathbf{c}} \geq \frac{\beta_{\rm SF}\sum_{\mathbf{w}_k \in \Phi_{k_0}}4\pi h_{\mathbf{w}_k} ||\mathbf{w}_k||^{-\alpha}\textbf{1}({\bf x} \in \mathcal{C})}{ R^{-2\alpha}(1-\beta_{\rm SF})}\; \bigg|\; \Phi_{k_0} \right) \nonumber\\
    &\hspace*{0cm}\overset{(a)}{=} \mathbb{E}_{h_{\mathbf{w}_k}} \!\!\left[
    \exp{\left(\frac{-\beta_{\rm SF}\sum_{\mathbf{w}_k \in \Phi_{k_0}}4\pi h_{\mathbf{w}_k} ||\mathbf{w}_k||^{-\alpha}\textbf{1}({\bf x} \in \mathcal{C})}{\bar{\sigma} R^{-2\alpha}(1-\beta_{\rm SF})} \right)}\right] \nonumber\\
    &\hspace*{0cm} = \prod_{\mathbf{w}_k \in \Phi_{k_0}} \!\!\!\!p \mathbb{E}_{h_{\mathbf{w}_k}} \left[\exp{\left(\frac{-\beta_{\rm SF} 4\pi \gamma P h_{\mathbf{w}_k} ||\mathbf{w}_k||^{-\alpha}}{\gamma \bar{\sigma}P R^{-2\alpha}(1-\beta_{\rm SF})} \right)}\right] + (1 - p) \nonumber\\
    &\overset{(b)}{=} \prod_{\mathbf{w}_k \in \Phi_{k_0}} \frac{p}{1 + \beta_{\rm SF}^\prime ||\mathbf{w}_k||^{-\alpha}} + (1 - p),
    \label{eq:eq_md2}
\end{align}
where $\beta_{\rm SF}^\prime = \frac{4\pi \beta_{\rm SF}}{\bar{\sigma}R^{-2\alpha}(1-\beta_{\rm SF})}$. Step (a) follows from the exponential distribution of $\sigma_{\mathbf{c}}$ having mean $\bar{\sigma}$, and step (b) follows by taking the expectation over fading $h_{\mathbf{w}}$.

Calculating the distribution of $p_{{\rm SF},k}(\beta_{\rm SF})$ directly is not feasible. Instead, the moments of the problem are calculated and subsequently converted into the distribution~\cite{wang2023fast}. Examining the instances of $p_{{\rm SF},k}(\beta_{\rm SF})$ offers a valuable understanding of the fluctuation and dispersion of effective detection of the target. For instance, it enables us to analyze the mean local delay, representing the expected number of transmissions required for successful detection. The following two lemma give the moments of $p_{{\rm SF},{\rm P}}(\beta_{\rm SF})$ and $p_{{\rm SF},{\rm B}}(\beta_{\rm SF})$.
\begin{lemma}
\label{le:lemm_md1}
PLCP: The $b$--th moment of $p_{{\rm SF},{\rm P}}(\beta_{\rm SF})$ where $b \in \mathbb{C}$ is given as 
\begin{align*}
    &M_{b,{\rm P}}(\beta_{\rm SF}) 
    \!=\! \exp\!\Bigg(-\lambda \int_{0}^{R_{\rm P}} 1 - \left(\frac{p}{1+\beta_{\rm SF}^\prime {v}_{\rm P}^{-\alpha}} + 1 - p\right)^b\!\! {\rm d}v_{\rm P} \\
    &\hspace*{2cm} -\lambda_{\rm L} \int_{\mathbb{R}^{+}} \int_0^{2\pi} 1 - \exp\bigg(-\lambda \int_{a_{\rm P}}^{b_{\rm P}} 1 - \\
    &\hspace*{1.5cm} \left(\frac{p}{1+\beta_{\rm SF}^\prime ||\mathbf{w}_{\rm P}||^{-\alpha}} + 1-p \right)^b \,{\rm d}v_{\rm P}\bigg){\rm d}\theta\, {\rm d}r\Bigg).
\end{align*}
where $||\mathbf{w}_{\rm P}||$ is distance to interfering radar. Also, $a_{\rm P}$, and $b_{\rm P}$ can be found using Lemma 2 in~\cite{shah2024modeling}.
\end{lemma}
\begin{proof}
From~\eqref{eq:eq_md2} we get,
\begin{align*}
    &M_{b,{\rm P}}(\beta_{\rm SF}) = \mathbb{E}_{\Phi_{{\rm P}_0}} \left[\prod_{\mathbf{w}_{\rm P} \in \Phi_{{\rm P}_0}} \frac{p}{1 + \beta_{\rm SF}^\prime ||\mathbf{w}_{\rm P}||^{-\alpha}} + (1 - p)\right] \\
    &\overset{(d)}{=} \exp\Bigg(-\lambda \int_{0}^{R_{\rm P}} 1 - \left(\frac{p}{1+\beta_{\rm SF}^\prime {v}_{\rm P}^{-\alpha}} + 1 - p\right)^b \,{\rm d}v_{\rm P} \\
    &\hspace*{2cm} -\lambda_{\rm L} \int_{\mathbb{R}^{+}} \int_0^{2\pi} 1 - \exp\bigg(-\lambda \int_{a_{\rm P}}^{b_{\rm P}} 1 - \\
    &\hspace*{1.5cm} \left(\frac{p}{1+\beta_{\rm SF}^\prime ||\mathbf{w}_{\rm P}||^{-\alpha}} + 1-p \right)^b \,{\rm d}v_{\rm P}\bigg){\rm d}\theta\, {\rm d}r\Bigg).
\end{align*}
Step (d) follows by applying the Laplace functional of the reduced \ac{PLCP} palm distribution (See Lemma 4.12~\cite{dhillon2020poisson}).
\end{proof}

\begin{lemma}
\label{le:lemm_md2}
BLCP: The $b$-th moment of $p_{{\rm SF},{\rm B}}(\beta_{\rm SF})$ where $b \in \mathbb{C}$ is given as 
\begin{align*}
    &M_{b,{\rm B}}(\beta_{\rm SF}) 
    = \exp\left(\!\!-\lambda \int_{0}^{R_{\rm B}}\!\!\! 1 - \left(\frac{p}{1+\beta_{\rm SF}^\prime {v}_{\rm B}^{-\alpha}} + 1 - p\right)^b \,{\rm d}v_{\rm B} \right)\\
    &\hspace*{0.5cm} \Bigg[\frac{1}{2\pi R_{\rm g}} \int_{0}^{R_{\rm g}} \int_0^{2\pi} \exp\bigg(-\lambda \int_{a_{\rm B}}^{b_{\rm B}} 1 - \\
    &\hspace*{1.2cm} \left(\frac{p}{1+\beta_{\rm SF}^\prime ||\mathbf{w}_{\rm B}||^{-\alpha}} + 1-p \right)^b \,{\rm d}v_{\rm B}\bigg){\rm d}\theta\, {\rm d}r\Bigg]^{n_{\rm B}-1}.
\end{align*}
where $||\mathbf{w}_{\rm B}||$ is distance to interfering radar. Also, $a_{\rm B}$, and $b_{\rm B}$ can be found using Theorem 1 in~\cite{shah2024modeling}.
\end{lemma}
\begin{proof}
From~\eqref{eq:eq_md2} we get,
\begin{align*}
    &M_{b,{\rm B}}(\beta_{\rm SF}) \!=\! \mathbb{E}_{\Phi_{{\rm B}_0}} \left[\prod_{\mathbf{w}_{\rm B} \in \Phi_{{\rm B}_0}} \frac{p}{1 + \beta_{\rm SF}^\prime ||\mathbf{w}_{\rm B}||^{-\alpha}} + (1 - p)\right] \\
    &= \exp\left(-\lambda \int_{0}^{R_{\rm B}} 1 - \left(\frac{p}{1+\beta_{\rm SF}^\prime {v}_{\rm B}^{-\alpha}} + 1 - p\right)^b \,{\rm d}v_{\rm B} \right)\\
    &\hspace*{0.5cm} \Bigg[\frac{1}{2\pi R_{\rm g}} \int_{0}^{R_{\rm g}} \int_0^{2\pi} \exp\bigg(-\lambda \int_{a_{\rm B}}^{b_{\rm B}} 1 - \\
    &\hspace*{1.2cm} \left(\frac{p}{1+\beta_{\rm SF}^\prime ||\mathbf{w}_{\rm B}||^{-\alpha}} + 1-p \right)^b \,{\rm d}v_{\rm B}\bigg){\rm d}\theta\, {\rm d}r\Bigg]^{n_{\rm B}-1}.
\end{align*}
\end{proof}

\begin{remark}
Lemma~\ref{le:lemm_md1} and~\ref{le:lemm_md2}, provide the moments of \ac{MD} of \ac{SF} for the \ac{PLCP} and \ac{BLCP} frameworks respectively. Here, $b=1$ represents the first moment, and $b=-1$ represents the first negative moments, both having their respective significance in characterizing the system performance. The first moment $M_{1,k}\left(\beta_{\rm SF} = \frac{\beta}{1+\beta}\right)$ gives the detection probability that an attempted detection by the ego radar of the target located at a distance $R$ is successful. Using the first moment, one can also determine the number of successful detections $n_{{\rm D},k}$ of vehicles for a fixed target distance $R$ (more details in Section~\ref{sec:Results}-A). Likewise the first negative moment $M_{-1,k}\left(\beta_{\rm SF} = \frac{\beta}{1+\beta}\right)$ gives the mean local delay $\mathcal{M}_k = \frac{1}{p} M_{-1,k} \left(\frac{\beta}{1+\beta}\right)$ and it represents the delay in detecting a target (more details in Section~\ref{sec:Results}-B).
\end{remark}

In Section~\ref{sec:Results}, we will discuss the insights gained by studying these moments, but first, we derive the detection probability of ego radar for \ac{PLCP} and \ac{BLCP} modeled automotive radar networks using the \ac{MD}.

\subsection{Detection Success Probability}
The detection \textit{success probability} at a threshold $\beta$ is defined as the \ac{CCDF} of SIR, $p_{{\rm D},Kk}(\beta) = \mathbb{P}[\rm{SIR} > \beta]$. This represents the probability that an attempted detection by the ego radar of the target located at a distance $R$ is successful.
\begin{corollary}
For the network where the locations of the vehicles are modeled as \ac{PLCP}, ($\Phi_{{\rm P}_0}$) or \ac{BLCP} ($\Phi_{{\rm B}_0}$) and $\beta_{\rm SF} = \frac{\beta}{1 + \beta}$, the first moment, $M_{1,k}(\beta_{\rm SF})$, is the detection success probability of an ego radar given by
\begin{align}
    & \text{PLCP:}\, M_{1,{\rm P}}\left(\frac{\beta}{1+\beta}\right) = p_{{\rm D},{\rm P}}(\beta) = \nonumber\\
    &\exp\Bigg(-\lambda \int_{0}^{R_{\rm P}} 1 - \left(\frac{p}{1+\beta^\prime {v}_{\rm P}^{-\alpha}} + 1 - p\right) \,{\rm d}v_{\rm P} \nonumber\\
    &\hspace*{1.5cm} -\lambda_{\rm L} \int_{\mathbb{R}^{+}} \int_0^{2\pi} 1 - \exp\bigg(-\lambda \int_{a_{\rm P}}^{b_{\rm P}} 1 - \nonumber\\
    &\hspace*{1cm} \left(\frac{p}{1+\beta^\prime ||\mathbf{w}_{\rm P}||^{-\alpha}} + 1-p \right) \,{\rm d}v_{\rm P}\bigg){\rm d}\theta\, {\rm d}r\Bigg).
    \label{eq:cor_1}
\end{align}
and
\begin{align}
    &\text{BLCP:}\, M_{1,{\rm B}}\left(\frac{\beta}{1+\beta}\right) = p_{{\rm D},{\rm B}}(\beta) = \nonumber\\
    &\exp\left(-\lambda \int_{R}^{R_{\rm B}} 1 - \left(\frac{p}{1+\beta^\prime {v}_{\rm B}^{-\alpha}} + 1 - p\right) \,{\rm d}v_{\rm B} \right)\nonumber\\
    &\hspace*{0.5cm} \Bigg[\frac{1}{2\pi R_{\rm g}} \int_{0}^{R_{\rm g}} \int_0^{2\pi} \exp\bigg(-\lambda \int_{a_{\rm B}}^{b_{\rm B}} 1 - \nonumber\\
    &\hspace*{1.5cm} \left(\frac{p}{1+\beta^\prime ||\mathbf{w}_{\rm B}||^{-\alpha}} + 1-p \right) \,{\rm d}v_{\rm B}\bigg){\rm d}\theta\, {\rm d}r\Bigg]^{n_{\rm B}}.
    \label{eq:cor_2}
\end{align}
where $\beta^\prime = \frac{4\pi\beta}{\bar{\sigma} R^{-2\alpha}}$, and $\beta$ is the \ac{SIR} threshold.
\begin{proof}
The proof of the corollary is straightforward and follows by taking $b=1$ in Lemma~\ref{le:lemm_md1} and~\ref{le:lemm_md2}.
\end{proof}
\end{corollary}

\subsection{Reconstruction of MD of $p_{{\rm SF},k}(\beta_{\rm SF})$ from $M_{b,k}(\beta_{\rm SF})$}
It is not feasible to compute the \ac{MD}s for most network models. In certain simple cases, the \ac{MD}s can be estimated by considering the closest interferer. 
Another way is by applying the Gil-Pelaez (GP) theorem~\cite{gil1951note} we compute the \ac{MD} using only imaginary moments as shown
    $F_{p_{\rm SF}}(z) = 0.5 + \frac{1}{\pi} \int_0^\infty \frac{\Im\left(e^{-ju \log(z)} M_{ju}\right)}{u} du$.
\begin{figure*}[t]
\centering
\subfloat[]
{\includegraphics[width=0.23\textwidth]{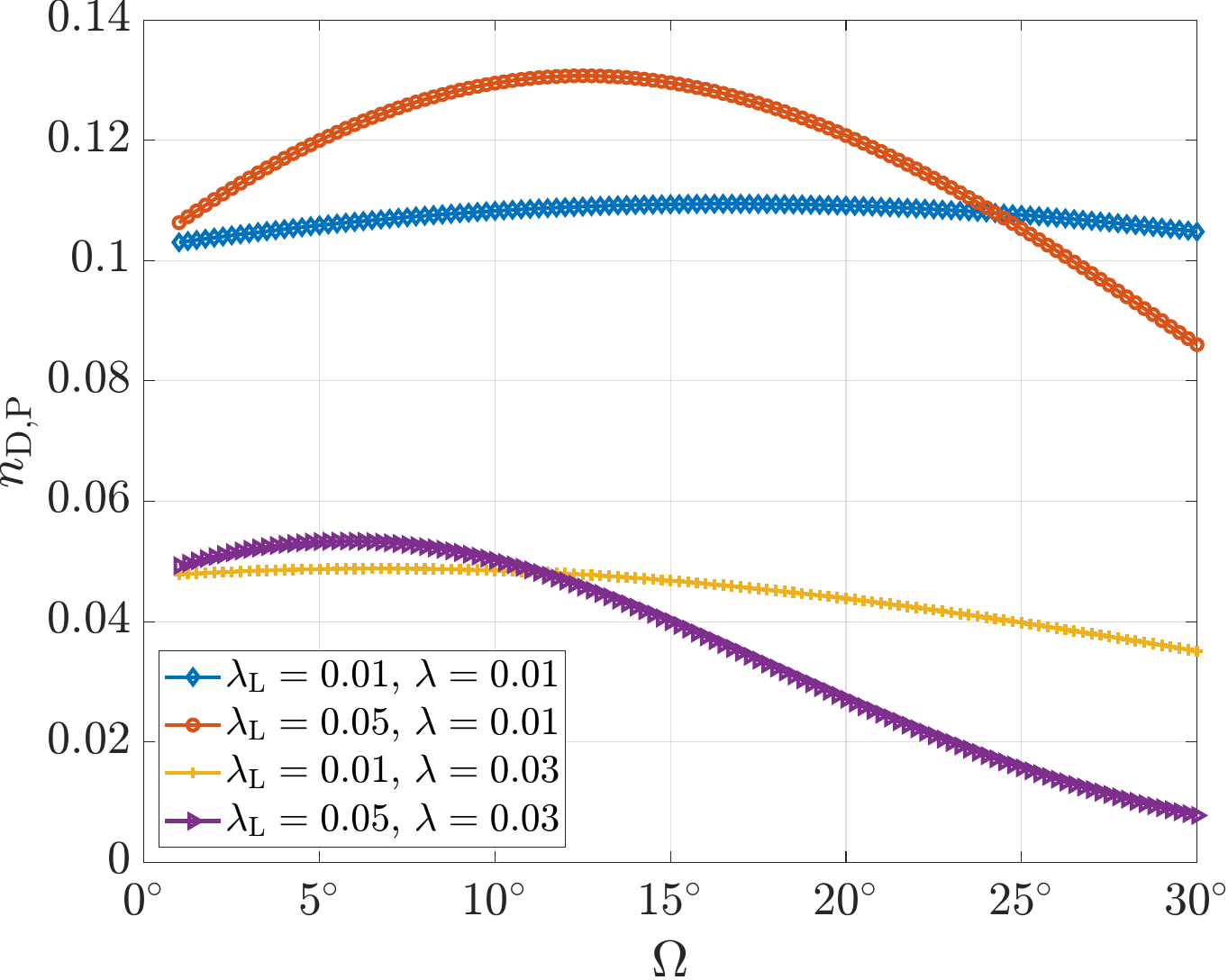}
\label{fig:result_P_1}}
\hfil
\subfloat[]
{\includegraphics[width=0.23\textwidth]{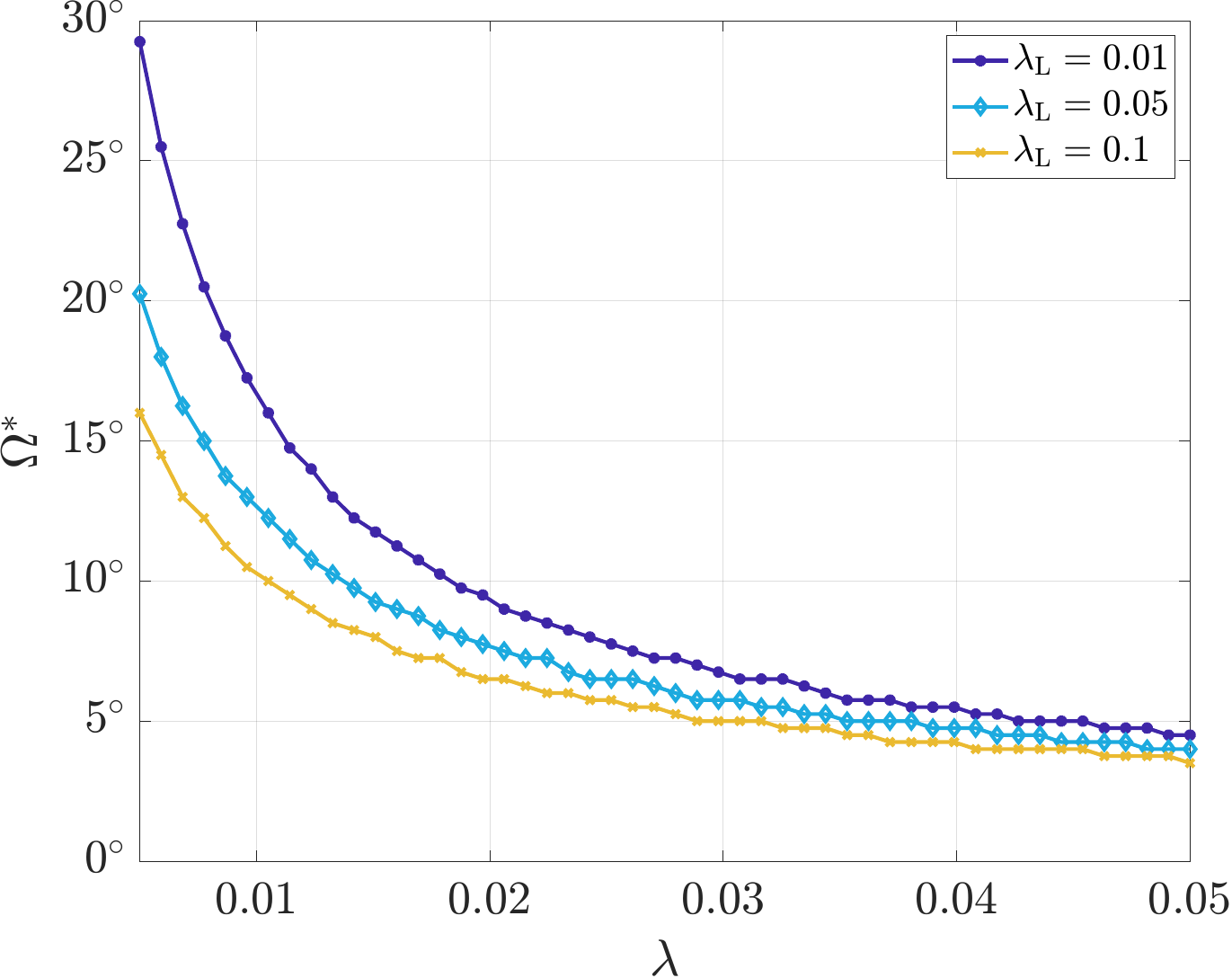}
\label{fig:resultOpt_P_1}}
\hfil
\subfloat[]
{\includegraphics[width=0.23\textwidth]{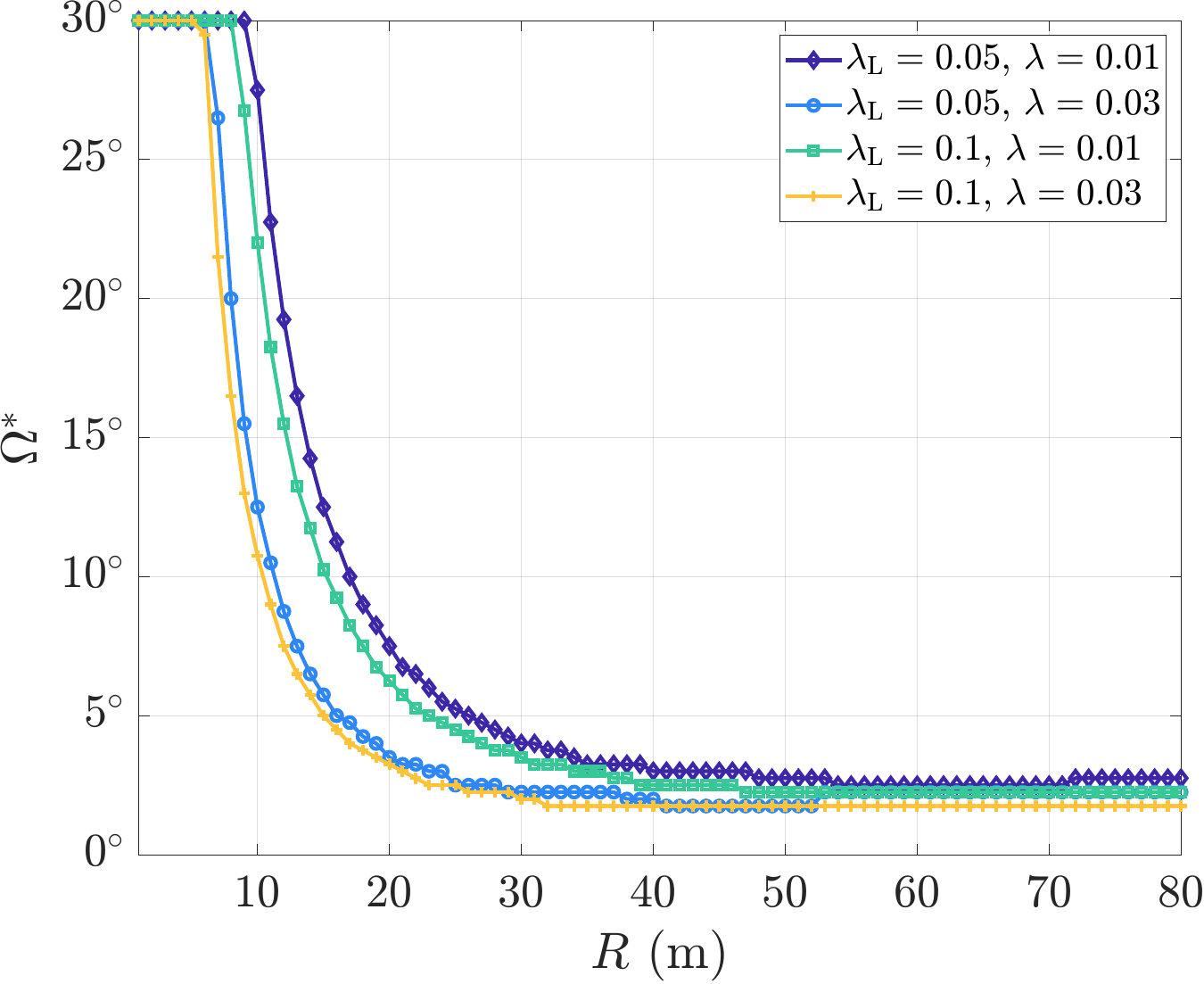}
\label{fig:resultOpt_P_2}}
\hfil
\subfloat[]
{\includegraphics[width=0.23\textwidth]{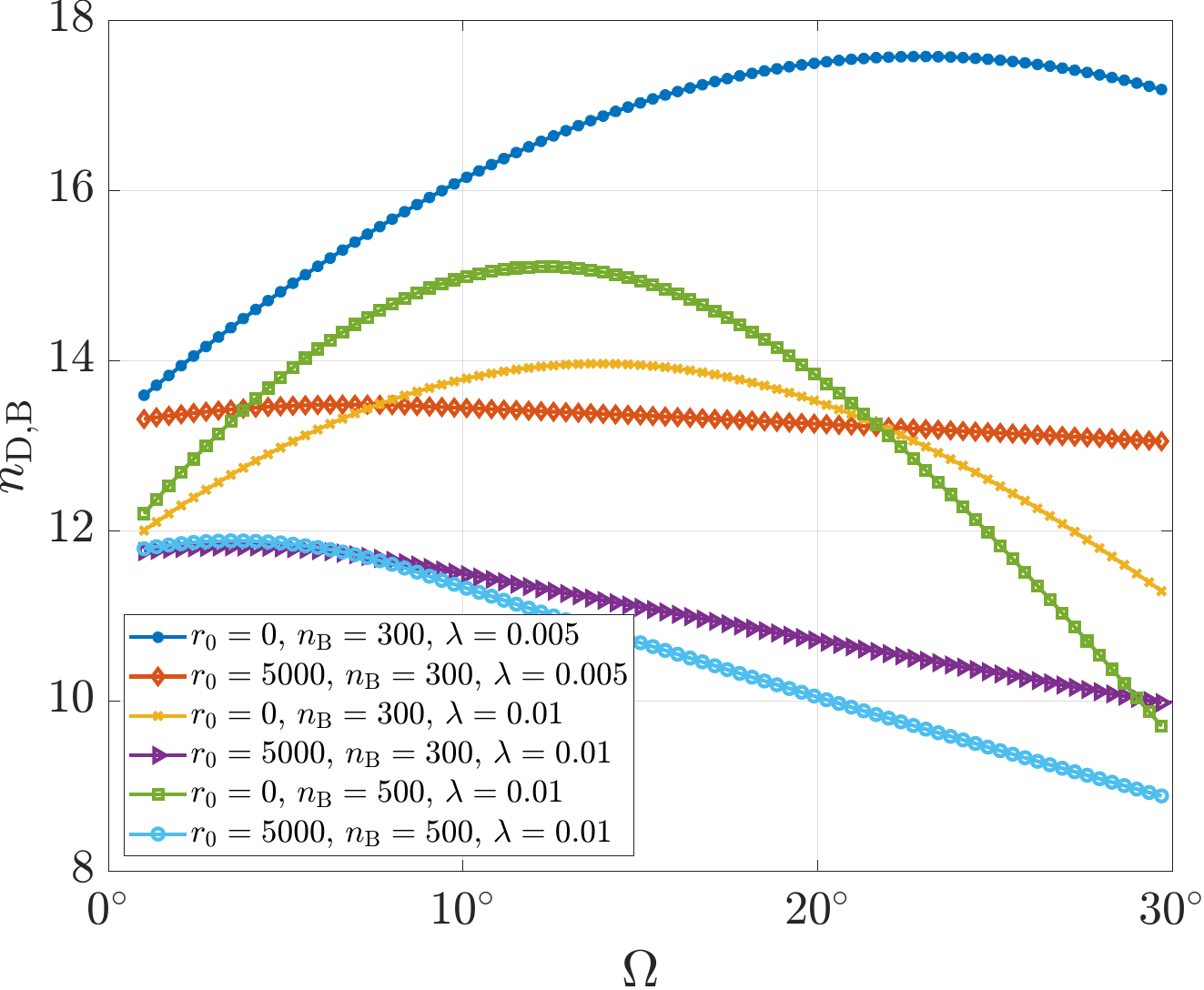}
\label{fig:result_B_1}}
\caption{(a) and (d) Number of successful detections versus $\Omega$ for \ac{PLCP} and \ac{BLCP} framework respectively, (b) and (c) Optimal beamwidth versus $\lambda$ and $R$ respectively for \ac{PLCP}.}
\label{fig:result1} 
\end{figure*}
The analytical evaluation of the integral is intractable in the case of \ac{SF} or \ac{SIR} \ac{MD}, and the numerical approximation computation of imaginary moments is computationally challenging. Instead, we use the Chebyshev-Markov (CM) method recently proposed by Wang \textit{et.al} in~\cite{wang2023fast}. Here, the authors reconstruct MD using a finite sequence of moments based on the Hausdorff moment problem (HMP). 
Hausdorff demonstrated that the existence of a distribution is contingent upon the infinite sequence of moments being completely monotonic. Furthermore, if the distribution does exist, it is unique. In practice, we have a finite sequence of moments, leading to a truncated HMP. Thus for the values of $M_{b}(\beta_{\rm SF})$ for $b$--th moments, we determine $F_{p_{\rm SF}}$ such that
\begin{align}
    \int_{0}^1 z^b dF_{p_{\rm SF}} (z) = M_{b}(\beta_{\rm SF})
    \label{eq:eq_sol}
\end{align}
for $b = \{1, \dots, n\}$ and $n \in \mathbb{N}$. If $\mathcal{F}_n$ is the set of all possible $F$ which solve~\eqref{eq:eq_sol} then the infimum and supremum of $F$ at the point of interest $x_0$, is given as $\inf_{F_{p_{\rm SF}} \in \mathcal{F}_n} F_{p_{\rm SF}}(x_0)$, and $\sup_{F_{p_{\rm SF}} \in \mathcal{F}_n} F_{p_{\rm SF}}(x_0)$.
The CM method consists of taking the average of infima and suprema to approximate the MD of $p_{\rm SF} (\beta_{\rm SF})$. If we take $n$ moments, the reconstructed MD, is 
$F_{p_{\rm SF}} (x, t_{\rm SF}) = \left( \inf_{F_{p_{\rm SF}} \in \mathcal{F}_n} F_{p_{\rm SF}}(x, t_{\rm SF}) + \sup_{F_{p_{\rm SF}} \in \mathcal{F}_n} F_{p_{\rm SF}}(x, t_{\rm SF})\right)$. 

Building upon prior work by Markov, their approach tackles the problem of finding the infimum and supremum of a distribution function $F_{p_{\rm SF}}$ within a specific range $[0, 1]$. The key idea is constructing a special discrete distribution $F_{p_{\rm SF}}^*$ with two important properties:
(1) Prescribed Moments: $F_{p_{\rm SF}}^*$ is designed to match the given finite sequence of moments $M_{b}(\beta_{\rm SF})$, and (2) Concentrated Mass: The maximum possible probability mass $p_0$ is concentrated at a specific point of interest $x_0$ within the range. This concentrated mass at $x_0$ essentially creates a \textit{worst-case scenario} distribution in terms of the range of $F_{p_{\rm SF}} (x_0)$. Mathematically, the infimum of $F_{p_{\rm SF}} (x_0)$ across all allowable distributions $\mathcal{F}_n$ is equal to $F_{p_{\rm SF}}^* (x_0^{-1})$. Conversely, the supremum is equal to $F_{p_{\rm SF}}^* (x_0)$. In our case, we already have the moments of $p_{\rm SF} (\beta_{\rm SF})$ given in Lemma~\ref{le:lemm_md1} and~\ref{le:lemm_md2}. Further, by leveraging the concept of jumps in a discrete distribution, where the jump locations $x_i$ represent possible values and the jump heights $p_i$ represent the corresponding probabilities. By carefully choosing these jumps and their heights, the constructed distribution $F_{p_{\rm SF}}^*$ achieves the desired properties of matching moments and concentrating mass at the point of interest.
The CM method provides a practical approach to approximate the \ac{MD} by constructing a discrete distribution with prescribed moments.
This approximation aligns with theoretical \ac{MD} characteristics, allowing for efficient reconstruction of complex distributions while adhering to computational feasibility.

\section{Numerical Results and Discussion}
\label{sec:Results}
This section highlights the results from a typical automotive radar \ac{PLCP} and \ac{BLCP} framework. The radar parameters used for the numerical results are based on typical automotive radars obtained from~\cite{series2014systems}. Specifically, $P = 10$ dBm, $\bar{\sigma} = 30$ dBsm, $\alpha = 2$, $G_{\rm t} = G_{\rm r} = 10$ dBi, $f_c = 76.5$ GHz, $\beta = 1$ dB (\ac{SIR} threshold), and $\beta_{\rm SF} = 0.5$ (\ac{SF} threshold). In all the plots, we assume $R=15\, {\rm m}$, $\lambda = 0.01\, {\rm m}^{-1}$, $R_{\rm g} = 1500\, {\rm m}$, $n_{\rm B} = 300$, and $R_k = 500\, {\rm m}$ unless specified otherwise.

\subsection{Optimal Parameters}
We denote the average number of successful detections, $n_{{\rm D},k}$, as the product of the probability of successful detection and the average number of potential targets present within the bounded radar beam of the ranging distance $R$, i.e.,
    $n_{{\rm D},k} \geq n_k (R) \times p_{{\rm D},k} = \lambda l_k  p_{{\rm D},k}.$
It is a measure for approximating the minimum number of successful detections, thus serving as a significant instrument for optimizing radar system design. Figure~\ref{fig:result_P_1} shows the plot of $n_{{\rm D},{\rm P}}$ with respect to beamwidth ($\Omega$).
From Remark~\ref{rem:remark_1}, $n_{\rm P} (R)$ has a linear relation with the intensity of the lines/streets and the intensity of vehicles on the road. Due to the symmetric structure of the radar sector and the homogeneous nature of \ac{PLCP}, the same behavior of linearity follows with beamwidth, i.e., $n_{\rm P} (R)$ increases linearly as $\Omega$ increases, while the detection probability decreases with increasing beamwidth~\cite{shah2024modeling}. As we take the product of $n_{\rm P} (R)$ and $p_{{\rm D},{\rm P}}$, $n_{{\rm D},{\rm P}}$ first increases as $n_{\rm P} (R)$ dominates $p_{{\rm D},{\rm P}}$, then after a certain critical value of $\Omega$, $p_{{\rm D},{\rm P}}$ dominates $n_{\rm P} (R)$, and $n_{{\rm D},{\rm P}}$ then decreases. As the intensity of lines and the beamwidth increase, the average length of line segments within the radar sector also increases. However, the increase in $\lambda_{\rm L}$ and beamwidth results in a reduced probability of successful detection due to increased interference. Consequently, a trade-off exists between $p_{{\rm D},{\rm P}}$ and the average length 
and an optimal beamwidth value, $\Omega^\ast$, exists, for which $n_{{\rm D},{\rm P}}$   is maximized  and likewise for $n_{{\rm D},{\rm B}}$, $\Omega^\ast= \arg\max n_{{\rm D},k}$. Fig.\ref{fig:result_B_1} shows the plot of $n_{{\rm D},{\rm B}}$ with respect to beamwidth and similar trade-off exists between $p_{{\rm D},{\rm B}}$ and $n(R)$. Unlike \ac{PLCP}, we find an optimal beamwidth, for which $n_{{\rm D},{\rm B}}$ is maximized as a function of the location of ego radar.

In Fig.~\ref{fig:resultOpt_P_1} we plot the optimal beamwidth $\Omega^\ast$ w.r.t $\lambda$ for 3 different values of $\lambda_{\rm L}$. The figure illustrates that the optimal beamwidth decreases with an increase in $\lambda$, eventually saturating. For a given value of $\lambda$, optimal beamwidth increases for increasing values of $\lambda_{\rm L}$.
Likewise, we find the optimal beamwidth as a function of distance to target $R$. 

Figure.~\ref{fig:resultOpt_P_2} illustrates the optimal beamwidth with respect to $R$ for various values of $\lambda_{\rm L}$ and $\lambda$. We see that for smaller values of $R$, irrespective of the intensity of streets and radars, the optimal beamwidth remains constant at $30^\circ$. For the smaller value of $\lambda = 0.01\,{\rm m}^{-1}$, optimal beamwidth remains constant for a larger span of $R$ as compared to the higher value of $\lambda$. 
Our analysis shows that the number of successful detections keeps increasing for smaller values of range $R$ with the increasing beamwidth. This is because with very few interferers present and a small $R$, even if we increase beamwidth, the effect of the interferers is not sufficiently significant to limit $n_{{\rm D},{\rm P}}$. Similarly when $\lambda_{\rm L} = 0.05 \,{\rm m}^{-2}$ and $\lambda = 0.01 \,{\rm m}^{-1}$, the optimal beamwidth is constant at $30^\circ$ for a larger range of $R$, as compared to $\lambda_{\rm L} = 0.1 \,{\rm m}^{-2}$ and $\lambda = 0.03 \,{\rm m}^{-1}$. Eventually, optimal beamwidth plateaus to some value. Further, we see that for higher intensity values of interfering radars, the optimal beamwidth plateaus are quicker compared to lower values. When $R$ is higher, the area of the radar sector is greater, including more interferers. Therefore, with the increase in $R$, the effect of interfering signal power received for higher $\lambda_{\rm L}$ and $\lambda$ causes the optimal beam to saturate to lower values quicker as compared to the scenario when $\lambda_{\rm L}$ and $\lambda$ are smaller.  
\begin{figure}
\centering
\subfloat[]
{\includegraphics[width=0.23\textwidth]{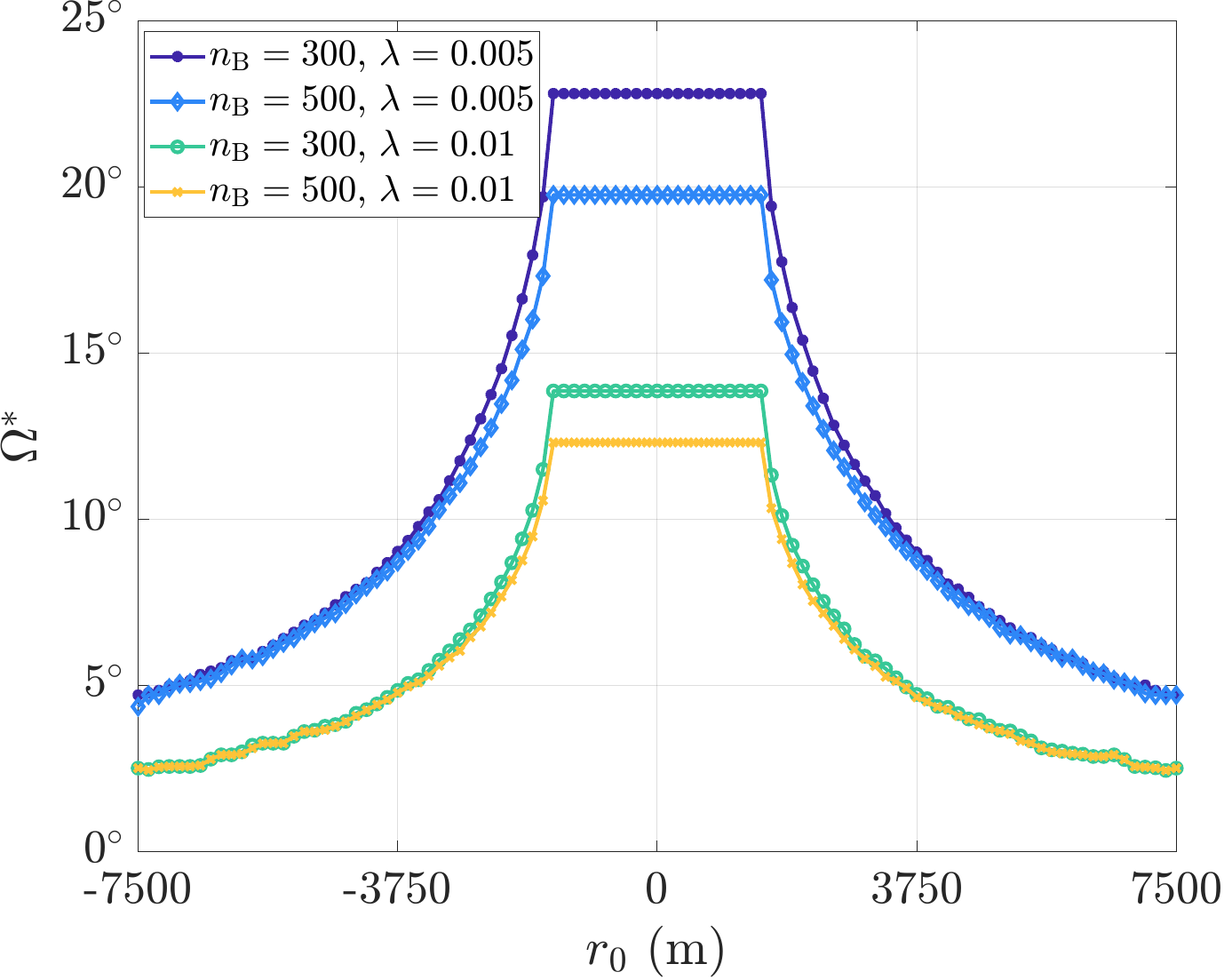}
\label{fig:resultOpt_B_1}}
\hfil
\subfloat[]
{\includegraphics[width=0.23\textwidth]{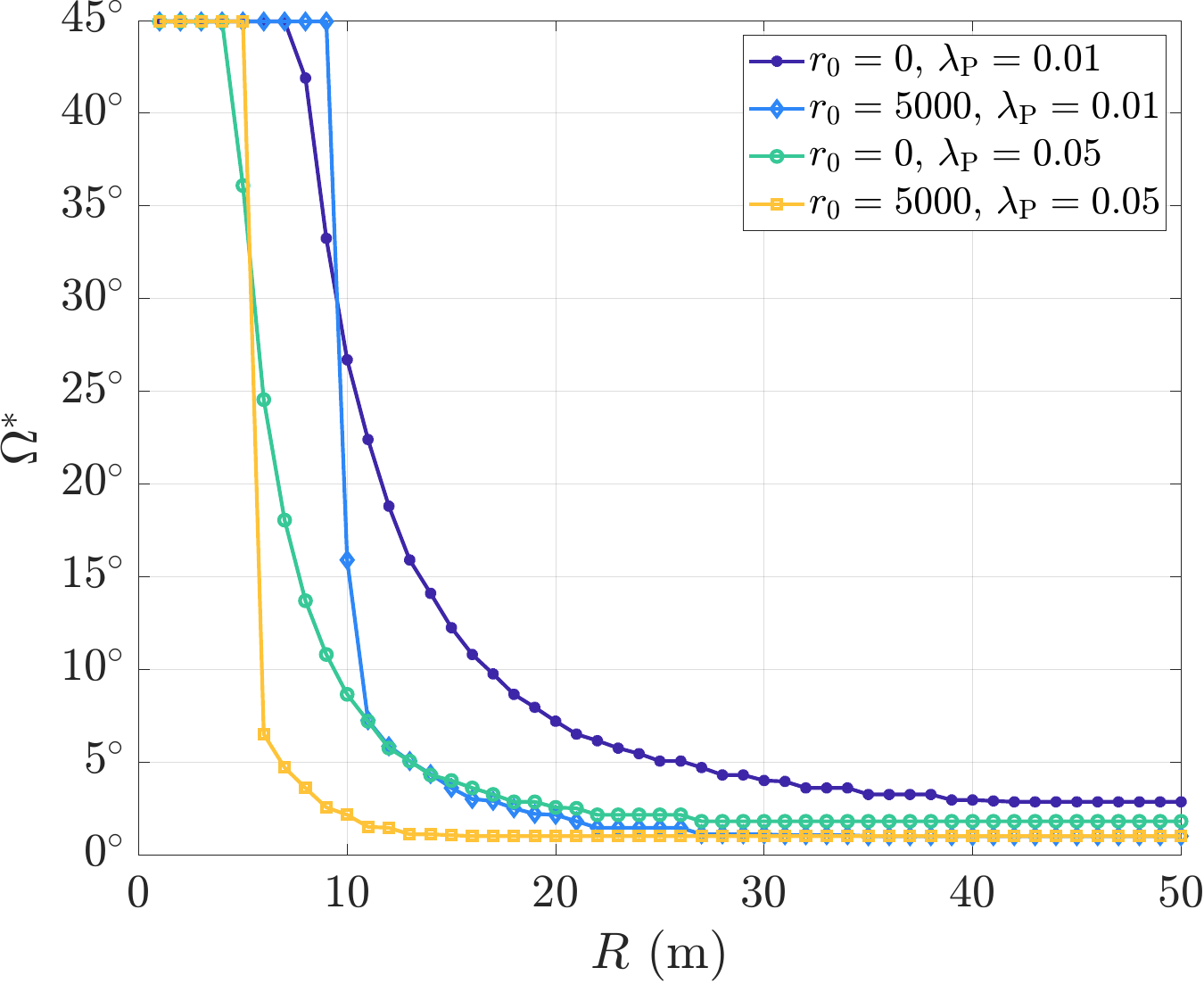}
\label{fig:resultOpt_B_2}}
\caption{(a) and (b) Optimal beamwidth versus $r_0$ and $R$ respectively for \ac{BLCP}.}
\end{figure}

Figures.~\ref{fig:result1}(e) and (f) illustrate the optimal beamwidth as a function of $r_0$, and $R$ respectively for parameters $n_{\rm B} = 300, 500$ and $\lambda = 0.005, 0.01\,{\rm m}^{-1}$. From a cognitive radar perspective, the optimum beamwidth for vehicles should be selected as the ego vehicle moves from one side of the city to the other side. 
Fig.~\ref{fig:resultOpt_B_1} illustrates that the optimal beamwidth first increases then saturates as the vehicle moves into the city, and then decreases as the vehicle exits the city center. 
Due to fewer interferers at the outskirts, we would anticipate that the ego radar can increase its beamwidth, here, to detect a larger number of targets. However, our analysis shows that we do not observe this in Fig.~\ref{fig:resultOpt_B_1}.
Instead, at the city's outskirts, having $\Omega > \Omega^{\ast}$ does not significantly increase the potential number of targets relative to the detection probability; thus, the optimal number of target detections occurs at smaller values of $\Omega$. The main insight we obtain from this result is that \emph{outside of the city center, the ego radar should prioritize monitoring vehicles traveling within its own lane rather than focusing excessively on those approaching from intersections.} For smaller values of $\lambda = 0.005\,{\rm m}^{-1}$, the ego radar will encounter fewer interferers; thus, the optimal beamwidth has values larger as compared to $\lambda = 0.01\,{\rm m}^{-1}$. We also see that at the outskirts of the city for a given $\lambda$, optimal beamwidth plateaus for $n_{\rm B} = 300$ and $500$. This is because the impact of interference from the intersecting streets decreases, and only the interference from the ego radars street affects $n_{\rm D}$.

Likewise, we find the optimal beamwidth w.r.t the distance to target $R$, as illustrated in Fig.~\ref{fig:resultOpt_B_2}. When the ego radar operates outside the city center and at initial values of $R < 10\, {\rm m}$, the optimal beamwidth remains fixed at $\Omega = 45^\circ$, indicating that the ego radar can maintain the largest possible beamwidth. As $R$ increases beyond $10\,{\rm m}$ with $\lambda = 0.01$, the ego radar should dynamically reduce its beamwidth from $45^\circ$ to $15^\circ$. The optimal beamwidth in the case of $r_0 = 5000$ plateaus quickly while for $r_0 = 0$, the ego radar has to adjust the beamwidth by increasing the amount of $R$ and plateauing slowly. As shown in Fig.~\ref{fig:resultOpt_B_2}, the saturated beamwidth for a fixed $\lambda$ is higher when $r_0 = 0$ than when $r_0 = 5000$, reflecting that at the outskirts, most potential targets are concentrated along $L_0$, whereas at $r_0 = 0$, additional targets appear due to interfering lines. Consequently, at $r_0 = 5000$, the ego radar selects a narrower beamwidth at larger values of $R$ to maintain a high detection probability.

\begin{figure}[t]
\centering
\subfloat[]
{\includegraphics[width=0.23\textwidth]{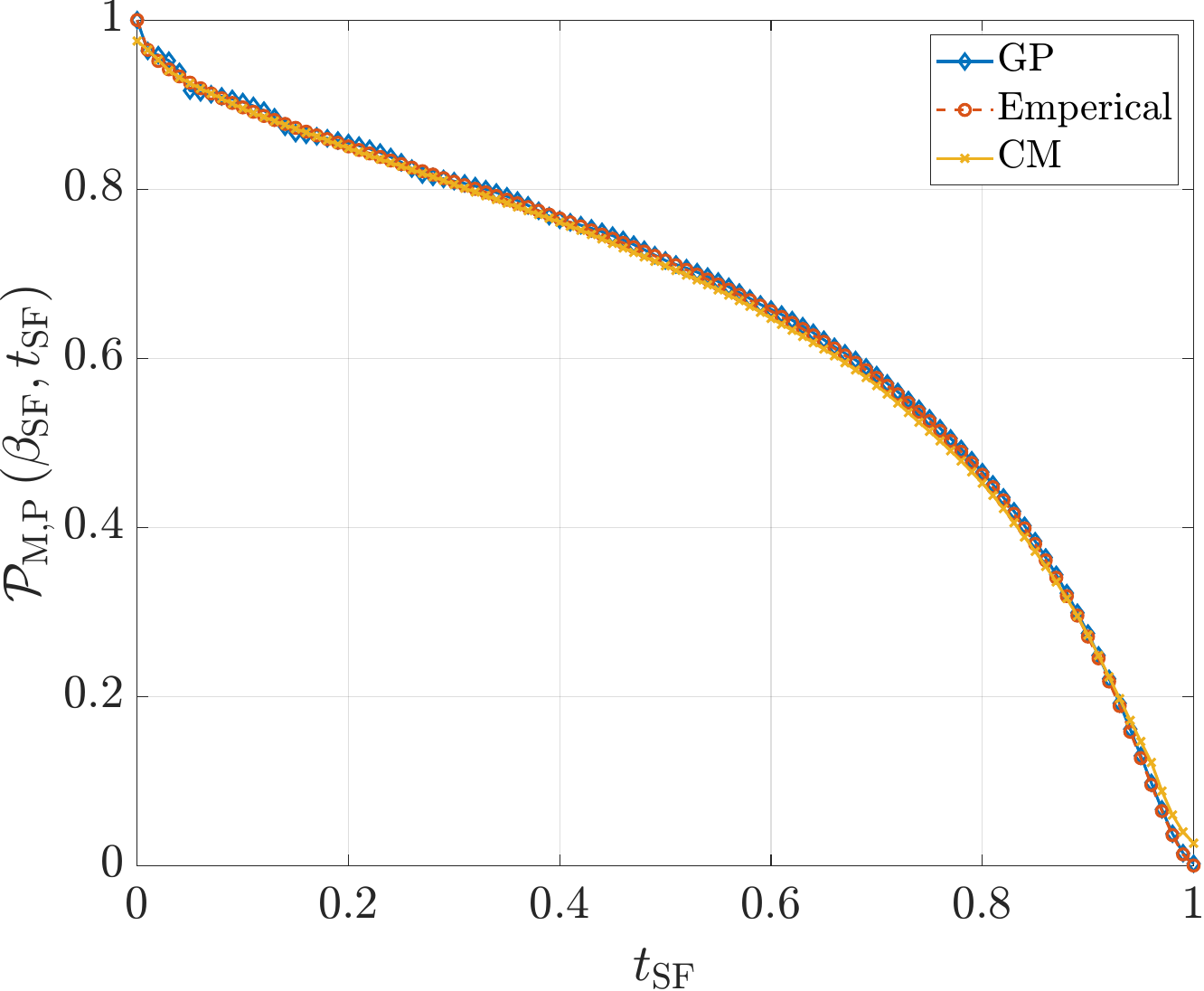}
\label{fig:resultMeta1}}
\hfil
\subfloat[]
{\includegraphics[width=0.24\textwidth]{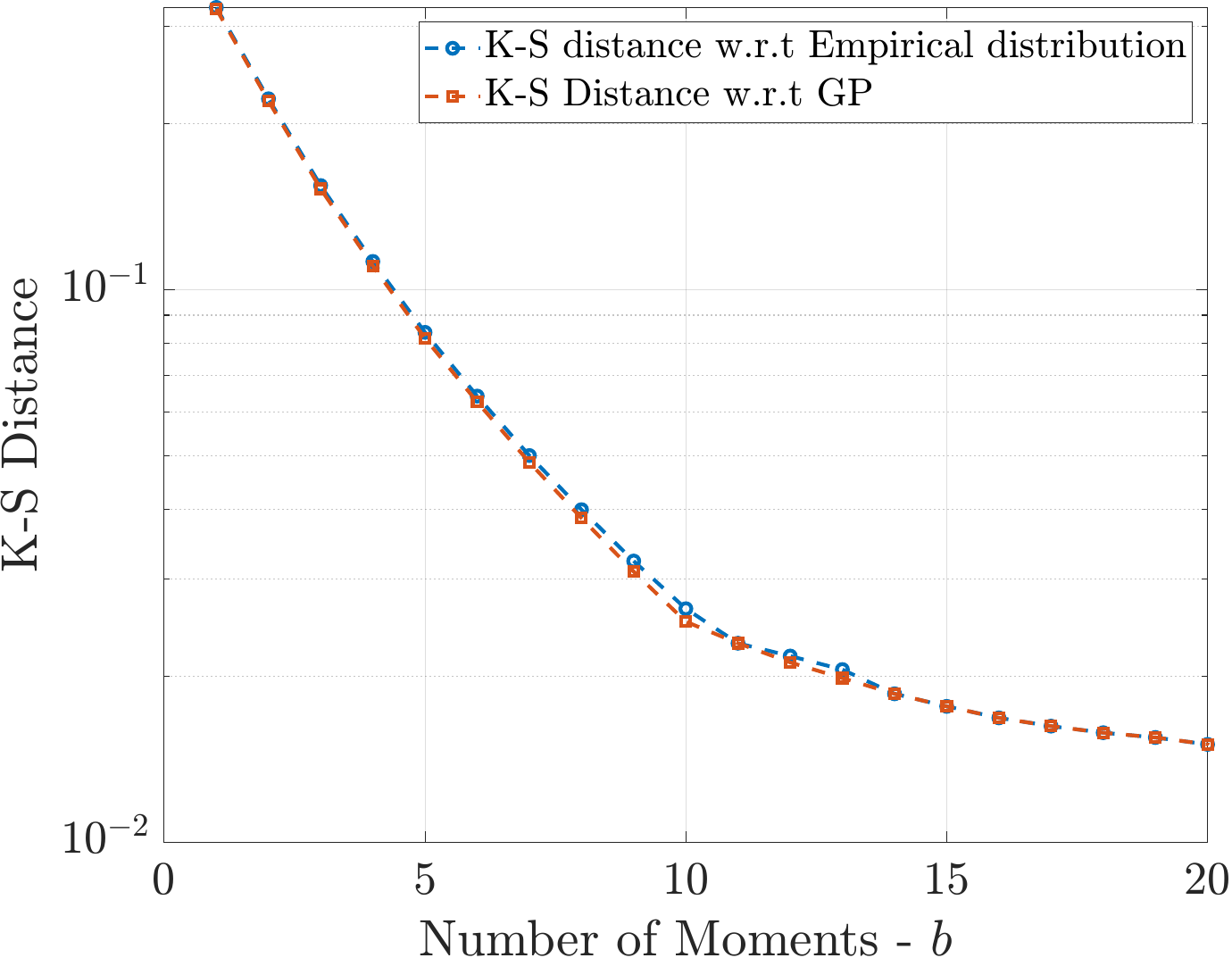}
\label{fig:resultMeta2}}
\caption{(a) \ac{SF} \ac{MD} of an ego radar in \ac{PLCP} framework generated through empirical, GP, and CM-bound methods, and (b) K-S distance between the CM-bound and empirical distribution, and the GP method, versus the number of moments}
\label{fig:resultMeta} 
\end{figure}

\begin{figure*}[t]
\centering
\subfloat[]
{\includegraphics[width=0.3\textwidth]{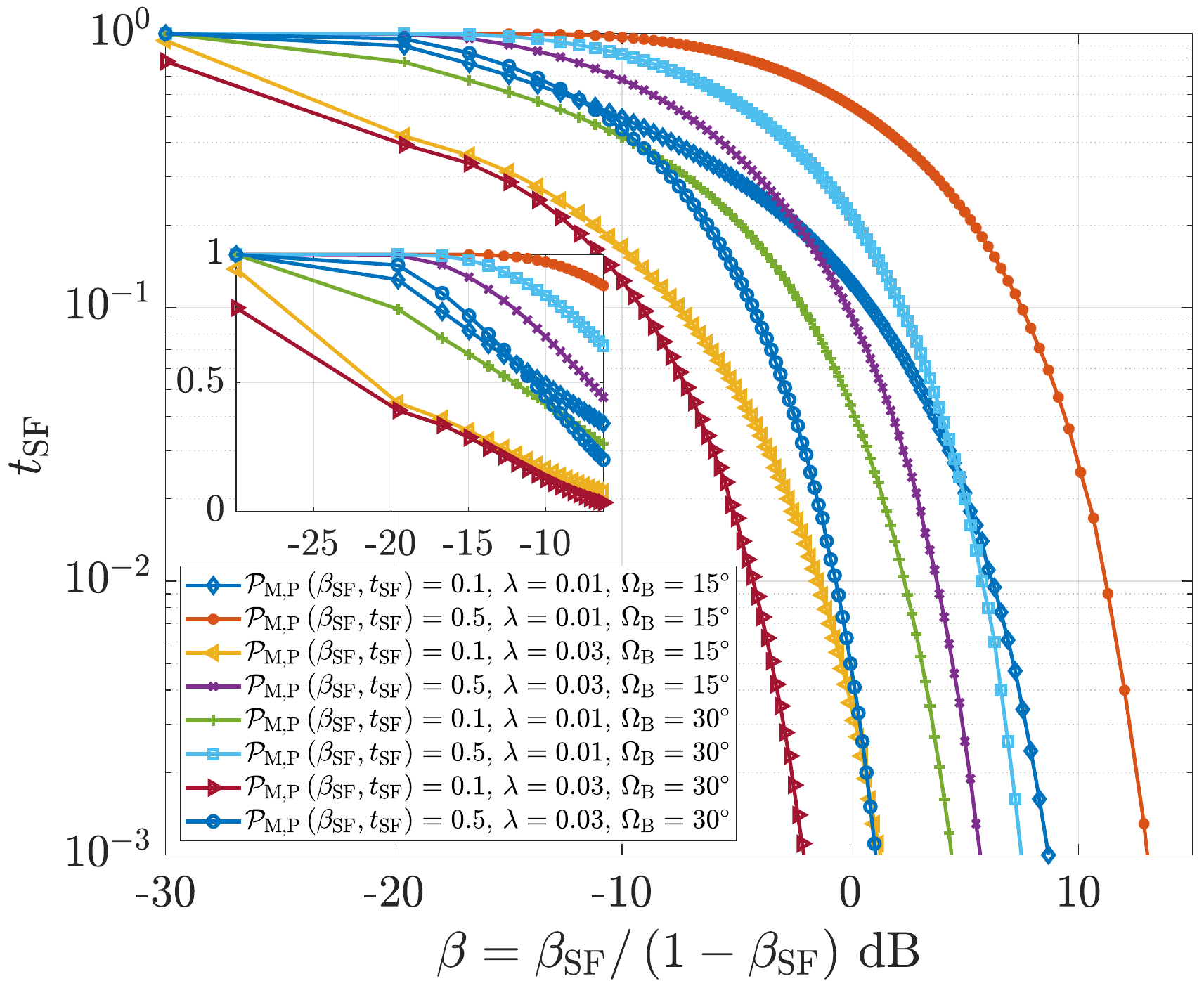}
\label{fig:P_md_res1}}
\hfil
\subfloat[]
{\includegraphics[width=0.3\textwidth]{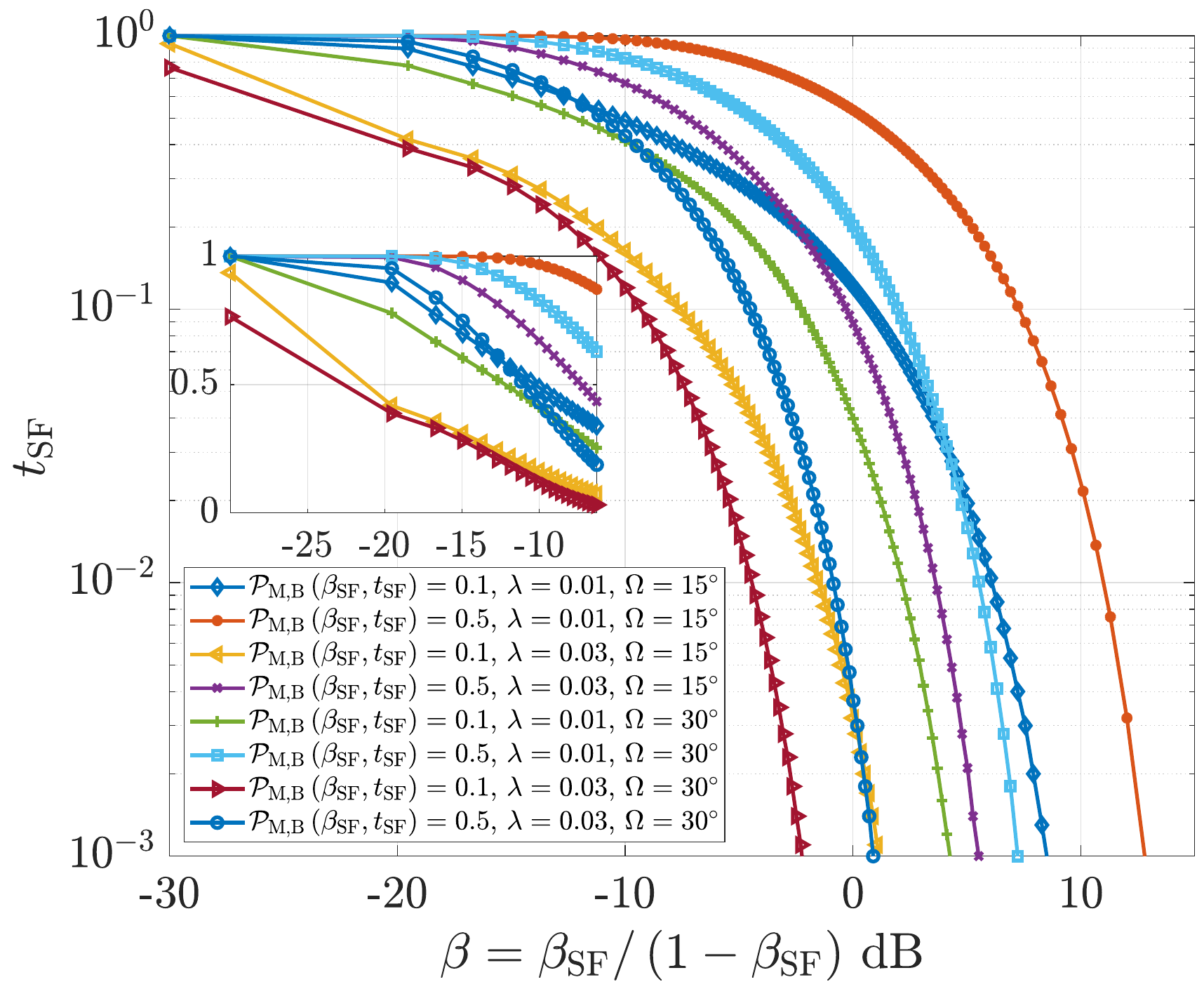}
\label{fig:B_md_res1}}
\hfil
\subfloat[]
{\includegraphics[width=0.3\textwidth]{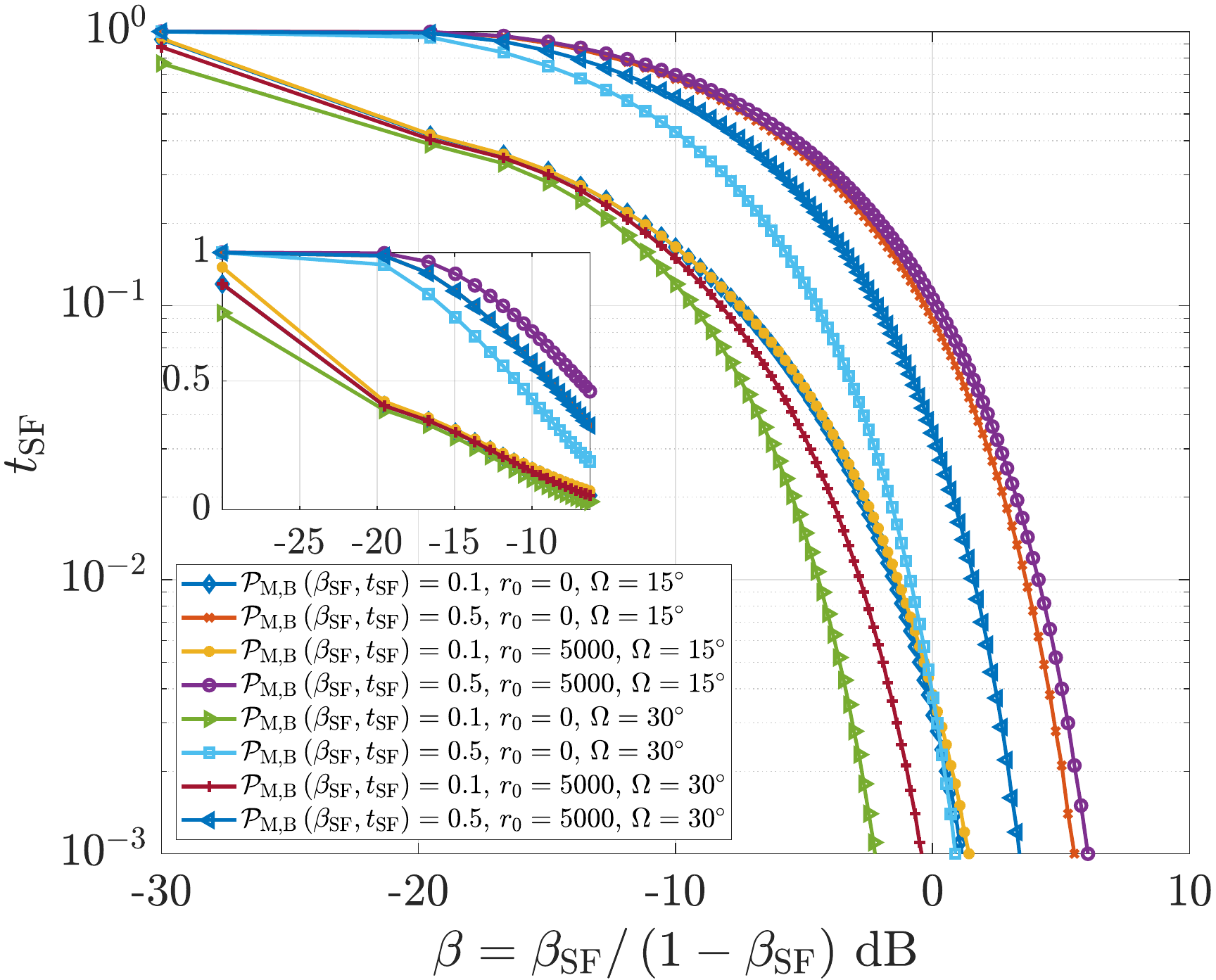}
\label{fig:B_md_res2}}
\caption{Plot between $t_{\rm SF}$ and $\beta_{\rm SF}$ where $\mathcal{P}_{{\rm M},k} \left(\beta_{\rm SF}, t_{\rm SF}\right) = \{0.1, 0.5\}$ for the CM method at $n = 21$ for (a) \ac{PLCP} framework, (b) \ac{BLCP} framework with $r_0 = 0$, and (c) for \ac{BLCP} framework for different values of $r_0$.}
\label{fig:resultMetaA} 
\end{figure*}

\subsection{Results on \ac{MD}}
\subsubsection{Reconstruction of \ac{MD}}
In section~\ref{sec:meta}, we noted that the \ac{MD} are infeasible to solve, and we rely on the novel CM bound method. To determine the fitness of MD reconstructed through the CM bound method, we compute the Kolmogorov–Smirnov (K-S) distance which is 
the largest absolute difference between the CDFs of the two samples, offering a measurable indication of the similarity between the distributions, with values ranging between 0 and 1. The bounded nature of this metric provides an advantage over other metrics, such as the Wasserstein distance and the Kullback-Leibler (KL) divergence, which can have a wide range of values and are not as easily interpretable. Compared to KL divergence, which measures the discrepancy between two probability distributions in terms of information loss, the KS distance has the benefit of being both symmetrical and non-parametric.

In Fig~\ref{fig:resultMeta1}, we plot the \ac{MD} of \ac{SF} in a \ac{PLCP} framework for values of $\beta_{\rm SF} = 10/(10+1)$, $\lambda_{\rm L} = \lambda = 0.01$, and $\Omega = 15^\circ$, and benchmark the reconstructed \ac{MD} using the CM bound method with the GP method, and the empirical distribution. From Fig.~\ref{fig:resultMeta1}, we see that for ten moments of $p_{\rm SF} \left(\beta_{\rm SF}\right)$, the CM bound methods achieve a construction error of $0.01$ with respect to the GP method and $0.005$ with respect to the empirical distribution, demonstrating a high level of accuracy. As seen in Fig.~\ref{fig:resultMeta1}, the reconstructed MD through GP method and empirical distribution are approximately equal to each other, thus we can any use any one the method to determine the fitness of MD generated through CM bound. Therefore, in Fig.~\ref{fig:resultMeta2}, we plot the K-S distance between the MD of CM bound w.r.t to both the GP and empirical methods. We take $b = \{1,\dots,21\}$ and calculate moments using $M_b \left(\beta_{\rm SF}\right)$, and then employ these moments in reconstructing MD. The resulting K-S distance decreases as $b$ increases, as shown in the plot. Any value of $b \geq 10$ offers a good match (i.e., the total distance is in $(0.005, 0.01]$). The only trade-off is that the computation time increases as we use more and more moments to reconstruct the MD. Thus, we conclude that the CM-bound method helps reconstruct MD with high accuracy. In all of the further results, we use the CM bound method to reconstruct the MD by taking 21 moments and using the reconstructed MD to derive results and insights.

\subsubsection{Inference from reconstructed MD - $\mathcal{P}_{{\rm M},k} \left(\beta_{\rm SF}, t_{\rm SF}\right)$} 
We illustrate how the reconstruction of the MD can improve network efficiency and eventually mitigate the effect of interference. In the context of an MD reliability analysis~\eqref{eq:eq_md}, $\mathcal{P}_{{\rm M},k} \left(\beta_{\rm SF}, t_{\rm SF}\right) = y$ represents the probability that a fraction of $1 - y$ of users achieve a reliability level of $t_{\rm SF}$ at the threshold SF value of $\beta_{\rm SF}$. By selecting the value of $\beta_{\rm SF}$, we determine the reliability $t_{\rm SF}$ corresponding to a specific SF value $\beta_{\rm SF}$, which represents the percentage of users who fail to achieve it. We identify the pairs $(\beta_{\rm SF}, t_{\rm SF})$ for $\mathcal{P}_{{\rm M},k} \left(\beta_{\rm SF}, t_{\rm SF}\right) = 0.1$ and $0.5$. This refers to the reliability of detecting a target and the SF threshold value at the $10$th percentile and $50$th percentile, indicating that $90\%$ and $50\%$ of users achieve or surpass this level of reliability, respectively. In our analysis, we assume $p = 1$, i.e., all radars are transmitting, and we vary $\beta_{\rm SF}$ from $0.001$ to $0.99$, which corresponds to $\beta$ in the range of $[-29.99, 18.47]$ dB. We first discuss the \ac{MD} results for \ac{PLCP}, followed by \ac{BLCP}.

{\bf PLCP:} For PLCP, we see that from Fig.~\ref{fig:P_md_res1}, at $\beta = 0$ dB, $\lambda = 0.01$, and $\Omega = 15^\circ$, $90\%$ of users can achieve detection probability with reliability level close to $0.125$. In contrast, $50\%$ of users can have reliability of $t_{\rm SF} = 0.547$ in detecting a target. As the beam width increases from $15^\circ$ to $30^\circ$, the reliability level for detecting a target decreases for all the users. Note that Fig.~\ref{fig:result_P_1} only illustrates an average view of the network, i.e., the number of successful detections averaged across all automotive radars, while  Fig.~\ref{fig:P_md_res1} offers insights about the deep network performance. For $\beta = -10$ dB and $\lambda = 0.03$, as $\Omega$ increases from $15^\circ$ to $30^\circ$, the reliability decreases from $t_{\rm SF} = 0.679$ to $0.447$ in case $\mathcal{P}_{{\rm M},{\rm P}} \left(\beta_{\rm SF}, t_{\rm SF}\right) = 0.5$, indicating that by increasing the beamwidth the reliability level for $50\%$ of users decreases at a faster rate.

Another interesting insight that can be obtained from this analysis in Fig.~\ref{fig:P_md_res1}, is that $50\%$ of users achieve the detection of a target with a reliability level of $0.97$ at $\lambda = 0.01$, $\Omega = 15^\circ$, and $\beta = -10$ dB. When we increase the intensity of vehicles to $0.03$, the reliability level falls to $t_{\rm SF} = 0.679$; if we increase the beam width to $30^\circ$, the reliability,$t_{\rm SF} = 0.84$. Thus, we get a higher reliability when the beamwidth of vehicles is increased as compared to if only the number of vehicles on streets increases. \emph{Thus, if we want to have half of the vehicles perform reasonably well in detecting several targets, increasing the beamwidth is beneficial for smaller $\lambda$, while in dense urban areas with high vehicular density, we will have a better performance if we keep beamwidth small.}

{\bf BLCP}: Next, we plot the \ac{MD} results for the \ac{BLCP} framework in Figs.~\ref{fig:resultMetaA}(b) and (c). In Fig.~\ref{fig:B_md_res1}, we plot the $t_{\rm SF}$ w.r.t $\beta$ for different values of $\mathcal{P}_{{\rm M},{\rm B}} \left(\beta_{\rm SF}, t_{\rm SF}\right)$, $\lambda$, and $\Omega$ when the ego radar is at the origin, i.e., $r_0 = 0$. 
Contrary to the PLCP, Fig.~\ref{fig:B_md_res2} shows how the reliability level changes as the location of the ego radar ($r_0$) varies across the city. For $\mathcal{P}_{{\rm M},{\rm B}} \left(\beta_{\rm SF}, t_{\rm SF}\right) = 0.1, 0.5$, and $\Omega = 15^\circ$, the reliability threshold is same irrespective of $r_0$. When the beamwidth is $15^\circ$, the interferers on the intersecting lines have a less prominent effect than $L_0$. Thus, the reliability threshold at different city locations is nearly constant as we have assumed a uniform  $\lambda$ across all lines. On the contrary, when $\Omega = 30^\circ$, at $\beta = -10$ dB, $90\%$ of users in the city center achieve detection probability with a reliability level close to $0.119$, while at the outskirts, i.e., $r_0 = 5000$ it is $0.148$. For $r_0 = 5000$, and $\Omega = 30^\circ$, $90\%$ of users have reliability level close to $0.4307$ which increases to $0.5766$ for $50\%$ of users. Thus, the intensity of vehicles, $\lambda$, and $\Omega$ are two parameters that can be optimized so that per radar detection performance can be optimized. 


\begin{figure*}[t]
\centering
\subfloat[]
{\includegraphics[width=0.23\textwidth]{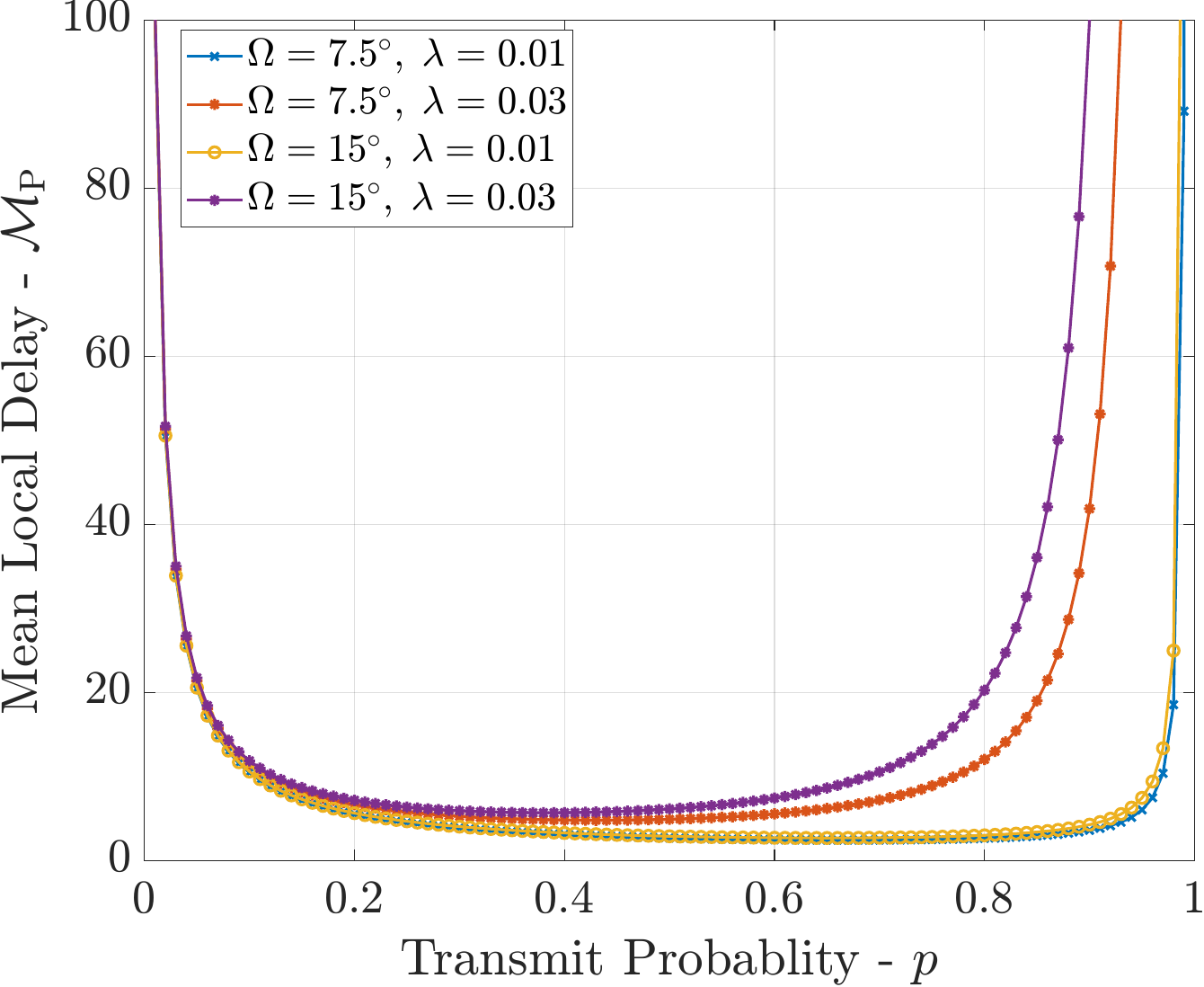}
\label{fig:delay_P_vs_p}}
\hfil
\subfloat[]
{\includegraphics[width=0.23\textwidth]{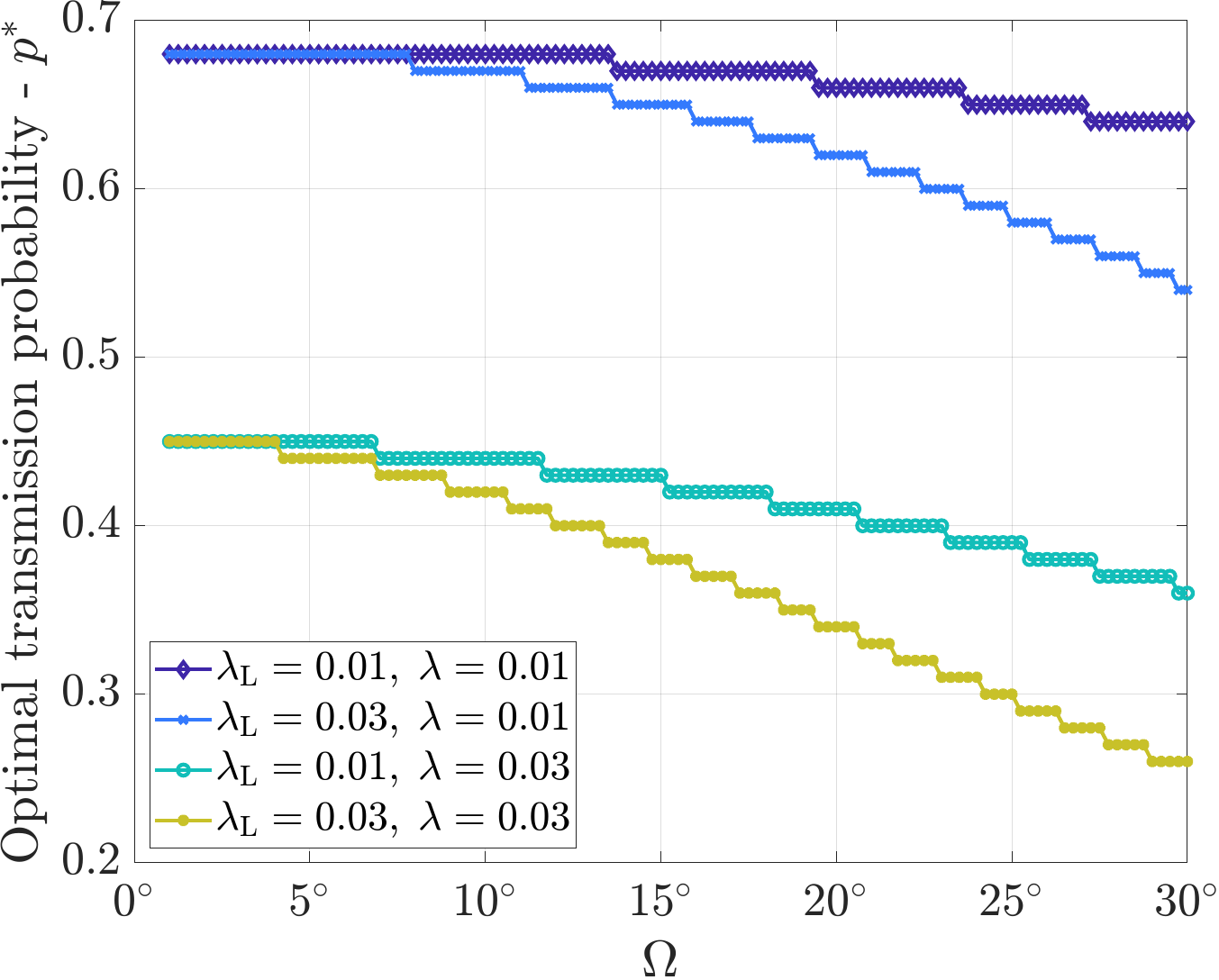}
\label{fig:P_opt_p_vs_omega}}
\hfil
\subfloat[]
{\includegraphics[width=0.23\textwidth]{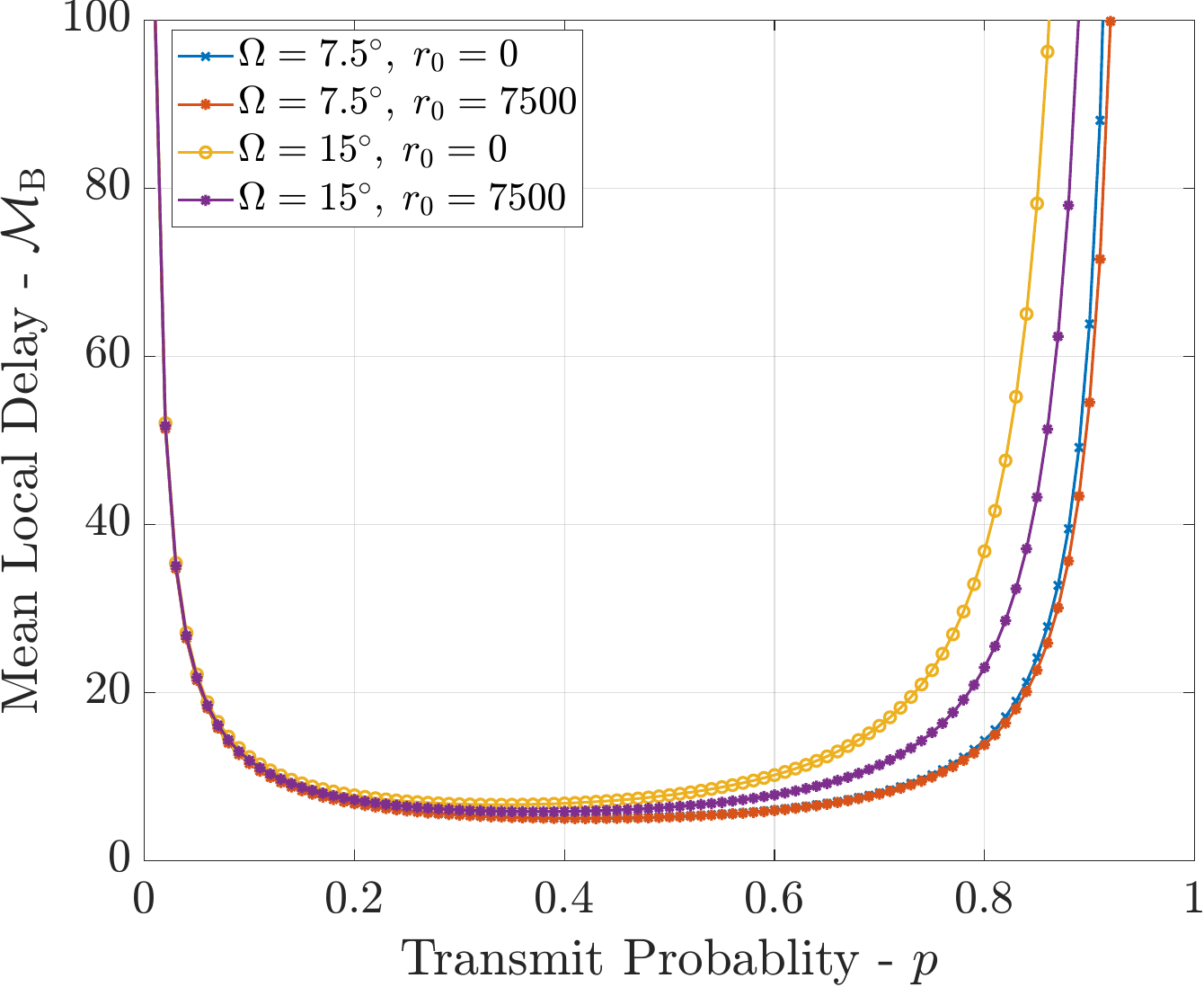}
\label{fig:delay_B_vs_p}}
\hfil
\subfloat[]
{\includegraphics[width=0.23\textwidth]{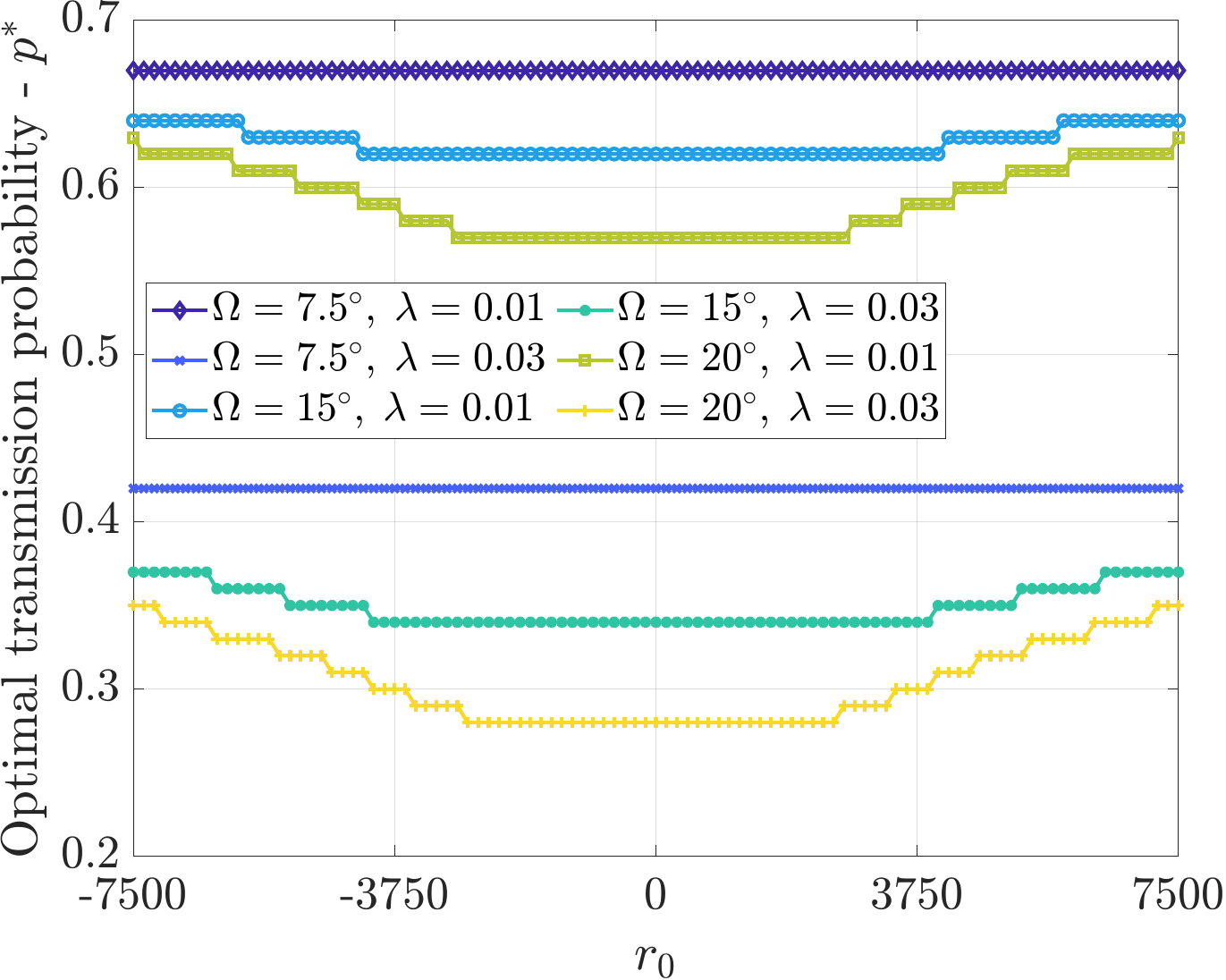}
\label{fig:B_opt_vs_vs_p}}
\caption{Transmission probability versus the mean local delay of an ego radar for (a) \ac{PLCP}, and (c) \ac{BLCP} model, (b) optimal transmission probability versus beamwidth in \ac{PLCP} model, and (d) optimal transmission probability versus $r_0$ for \ac{BLCP}.}
\label{fig:resultMetaB} 
\end{figure*}

\subsubsection{Optimization through MD}
In this section, we plot the metric which defines the delay in successfully detecting the target i.e., the mean local delay given as: $\mathcal{M}_k = \frac{1}{p} M_{{-1},k} \left(\beta_{\rm SF} = \frac{\beta}{1+\beta}\right)$. 
Consider a Bernoulli distribution; the first negative moment of this Bernoulli random variable represents the mean number of counts to get the first detection success. Therefore, if a radar transmits pulses with a probability $p$, then $\mathcal{M}_k$ represents the pulse delay in detecting a target. From Fig.~\ref{fig:delay_P_vs_p}, we observe that as the transmit probability increases, the local delay initially decreases because the transmitter tries to detect at a more frequent pace. However, a further increase in the transmit probability leads to a higher intensity of interfering automotive radar, which degrades detection probability and increases delay. This demonstrates that we can optimize the transmit probability to minimize the overall detection delay. Deriving this optimal probability is complex due to two opposing factors: increasing the transmit probability increases the frequency of detection attempts. However, this increases interference, reducing the detection success.

Figure.~\ref{fig:P_opt_p_vs_omega} shows the optimal transmit probability for minimizing the mean local delay in detecting a target w.r.t to the beamwidth of ego radar across various values $\lambda_{\rm L}$ and $\lambda$. 
As $\Omega$ increases, the optimal transmit probability $p^\ast$ initially remains constant and then declines. The rate by which the optimal transmission probability decreases is higher for larger intensity values of streets and vehicles. For $\lambda_{\rm L} = 0.01$ and $\lambda = 0.01, 0.03$, i.e., a sub-urban area, the optimal probability is greater than $0.5$ for all the values of $\Omega$, while as for the remaining values of $\lambda_{\rm L}$ and $\lambda$, optimal probability decreases as the beam width increases. Urban and dense urban areas show sensitivity to even slight variations of $\Omega$.

Likewise for the \ac{BLCP} framework,  we plot the mean local delay $\mathcal{M}_{\rm B} = \frac{1}{p} M_{{-1},{\rm B}} \left(\beta_{\rm SF}\right)$ for a ego radar w.r.t $p$ in Fig.~\ref{fig:delay_B_vs_p}. We plot $\mathcal{M}_{\rm B}$ for two different value of $\Omega = 7.5^\circ, 15^\circ$, and $r_0 = 0, 7500$. We see that an optimal transmit probability exists for which the delay is minimized; we also observe that at the city center, the ego radar experiences a higher delay in successful target detection as compared to the outskirts of $r_0 = 7500$. Leveraging the data from Fig.~\ref{fig:delay_B_vs_p}, we plot the optimal transmit probability $p^\ast$ w.r.t location of ego radar in Fig.~\ref{fig:B_opt_vs_vs_p}. For network engineers, it is favorable to know the transmission probability for a cognitive automotive radar so as to maximize the overall network performance. Figure.~\ref{fig:B_opt_vs_vs_p}, illustrates that for $\Omega = 7.5^\circ$, $p^\ast$ remains constant at all value of $r_0$ irrespective of $\lambda$. For a small $\Omega$, the ego radar experiences the same interference pattern across all locations as the interference contribution only comes due to interference present on line $L_0$. While for $\Omega = 15^\circ, 20^\circ$ the $p^\ast$ first decreases, then plateaus, and then again increases. The saturation region depends on the value of $\Omega$ and not $\lambda$. 
When ego radar operates outside the city, it may exhibit a higher $p$ due to reduced interferences compared to its operation within the city center.

\section{Conclusion}
\label{sec:con}
We presented a comprehensive analysis of automotive radar performance under the frameworks of both \ac{PLCP} and \ac{BLCP}, yielding novel insights into optimizing automotive radar parameters. Notably, we derived the optimal beamwidth and transmission probability parameters w.r.t to the dynamic location of ego radar within and outside the city. These insights, achieved through our \ac{BLCP} modeling approach, can enable cognitive radar systems to adapt and optimize detection capabilities in high-interference settings dynamically. We also demonstrated that the CM-bound method effectively reconstructs moment distributions with high accuracy, allowing for improved network efficiency and interference mitigation. By assuming a cognitively enabled vehicular radar network, we optimized the transmission frequency of the automotive radars so as to enhance target detection accuracy. 
This optimization allows for refined network efficiency and improved interference mitigation.
These conclusions suggest practical frameworks for optimizing radar designs in urban and suburban areas with high traffic variability.

\bibliographystyle{ieeetr}
\bibliography{references2}


\appendices

\section{}
\label{pr:thm_2}

The proof of the above theorem follows by integrating Lemma~\ref{le:radial_d}, where the area of integration and limits of integration depend upon the location of ego radar and the value of $R$. Let us consider the following cases.\\
\textbf{Case 1: $r_0 \in [-R_{\rm g}, R_{\rm g} - R]$.} Here, the radar sector is completely inside the city center, i.e., within the circle of radius $R_{\rm g}$. In Fig.~\ref{fig:th_3}, case 1 is marked, and the radar sector is within the city center. The radar sector is inside only if $-R_{\rm g} \leq r_0 \leq R_{\rm g} - R$. Now, the average length of line segments inside this radar sector is given by integrating $\rho(r)$ over the blue area as seen in case 1 of Fig.~\ref{fig:th_3} and is given as
\begin{align*}
    2\left(\int_{0}^{R\sin\Omega} \int_{mx+r_0}^{\sqrt{R^2 - x^2}+r_0} \frac{n_{\rm B}}{2 R_{\rm g}} {\rm d}y {\rm d}x\right) = \frac{n_{\rm B}}{2 R_{\rm g}} \Omega R^2.
\end{align*}
\begin{figure}[t]
    \centering
    \includegraphics[trim={0cm 0.7cm 0cm 0.5cm},clip,width = 0.35\textwidth]{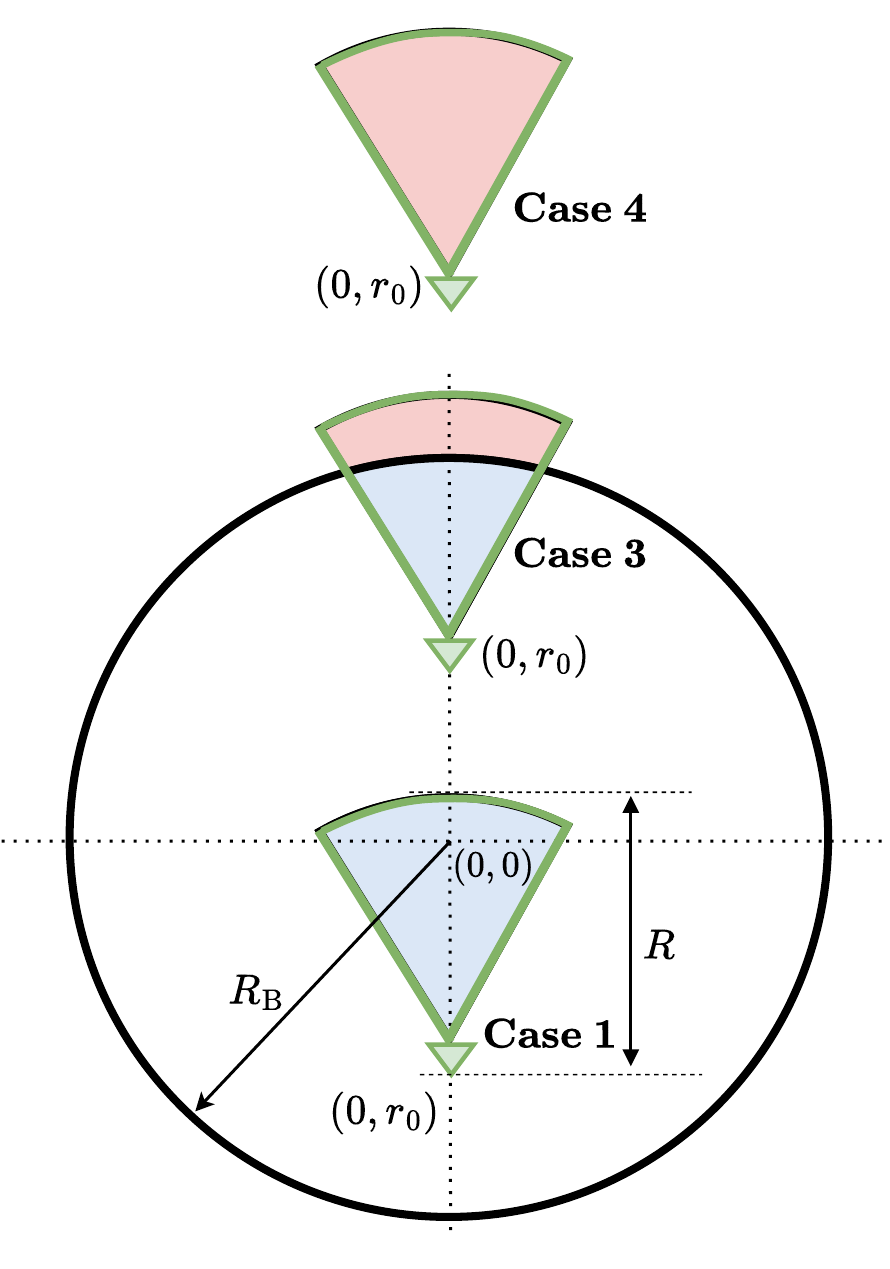}
    \caption{Different cases of the average length of lines in \ac{BLP} depending on the location of ego radar.}
    \label{fig:th_3}
\end{figure}
We will skip Case 2 and address it after Case 4.

\textbf{Case 3: $r_0 \in \left(\sqrt{R_{\rm g}^2 - \left(R \sin \Omega \right)^2} - R\cos\Omega, R_{\rm g}\right]$.} Here, the radar sector is partially inside and partially outside, as seen in Fig.~\ref{fig:th_3}. The circular arc part is outside, while the ego radar itself is present inside. This scenario occurs only if $r_0 > \sqrt{R_{\rm g}^2 - \left(R \sin \Omega \right)^2} - R\cos\Omega$. If the value of $r_0$ lies between $R_{\rm g} - R$, and $\sqrt{R_{\rm g}^2 - \left(R \sin \Omega \right)^2} - R\cos\Omega$ the circular arc of the radar sector will intersect with the circular arc of the disk of radius $R_{\rm g}$, which requires some careful analysis as we will see in Case 2. Now, in order to find $l_{\rm B}$ in this case we first integrate $\frac{n_{\rm B}}{2 R_{\rm g}}$ over the blue area which gives
\begin{align*}
    2\left(\int_0^{\rm y_C} \int_{mx+r_0}^{\sqrt{R_{\rm g}^2 - x^2}} \frac{n_{\rm B}}{2 R_{\rm g}} {\rm d}y {\rm d}x \right) &= \frac{n_{\rm B}}{2R_{\rm g}}\bigg(R_{\rm g}^2 \arcsin{\left(\frac{\rm y_C}{R_{\rm g}}\right)} \\
    &\hspace*{-2cm} + {\rm y_C}\sqrt{R_{\rm g}^2 - {\rm y^2_C}} - m{\rm y^2_C} - 2r_0{\rm y_C}\bigg).
\end{align*}
Where ${\rm y_C}=\frac{-m r_0 + \sqrt{R_{\rm g}^2(1+m^2) - r_0^2}}{1+m^2}$ which is the intersection point between ${\rm y_C} = mx + r_0$, and $x^2+{\rm y^2_C} = R_{\rm g}^2$ i.e. one of the lines of radar sector and disk of radius $R_{\rm g}$. Following this integration, we also need to find the area of $\frac{n_{\rm B}}{\pi R_{\rm g}} \arcsin{\left(\frac{R_{\rm g}}{\sqrt{x^2 + y^2}}\right)}$ over the red area, which is given as,
\begin{align*}
    &2\bigg(\int_0^{\rm y_C} \int_{\sqrt{R_{\rm g}^2 - x^2}}^{\sqrt{R^2 - x^2} + r_0} \frac{n_{\rm B}}{\pi R_{\rm g}} \arcsin{\left(\frac{R_{\rm g}}{\sqrt{x^2 + y^2}}\right)} {\rm d}y {\rm d}x + \\
    &\int_{\rm y_C}^{\rm y_B} \int_{mx+r_0}^{\sqrt{R^2 - x^2} + r_0} \frac{n_{\rm B}}{\pi R_{\rm g}} \arcsin{\left(\frac{R_{\rm g}}{\sqrt{x^2 + y^2}}\right)} {\rm d}y {\rm d}x \bigg).
\end{align*}

\textbf{Case 4: $r_0 \in \big\{(R_{\rm g},\infty) \cup (-\infty,-(R_{\rm g}+R)) \big\}$.} Here the radar sector is completely outside the disk $\mathcal{C}\left((0,0), R_{\rm g}\right)$. Fig.~\ref{fig:th_3} shows this scenario where the ego radar along with the radar sector is outside $\mathcal{C}\left((0,0), R_{\rm g}\right)$. In order to find $l_{\rm B}$, we determine the area of $\frac{n_{\rm B}}{\pi R_{\rm g}} \arcsin{\left(\frac{R_{\rm g}}{\sqrt{x^2 + y^2}}\right)}$ over the red are only, which gives $l_{\rm B}$ as,
\begin{align*}
    2 \int_0^{\rm y_B} \int_{mx+r_0}^{\sqrt{R^2 - x^2} + r_0} \frac{n_{\rm B}}{\pi R_{\rm g}} \arcsin{\left(\frac{R_{\rm g}}{\sqrt{x^2 + y^2}}\right)} {\rm d}y {\rm d}x.
\end{align*}

\textbf{Case 2} \& \textbf{6:} In case 3 we saw that if $\sqrt{R_{\rm g}^2 - \left(R \sin \Omega \right)^2} - R\cos\Omega < r_0 \leq R_{\rm g} - R$, the circular arc of radar sector will intersect with the circular arc of $\mathcal{C}\left((0,0), R_{\rm g}\right)$. This scenario gives rise to Case 2, where the ego radar is present within the disk $\mathcal{C}\left((0,0), R_{\rm g}\right)$, but the two circular arcs intersect. Likewise for case 5, if $-(R_{\rm g}+R) \leq r_0 < -\sqrt{R_{\rm g}^2 - \left(R \sin \Omega \right)^2} - R\cos\Omega$ the two circular arcs will again intersect, but here the ego radar is outside the disk. By first determining ${\rm y_A}$, i.e., the point of intersection between the $x^2 + {\rm y^2_A} = R_{\rm g}^2$, and $x^2 + ({\rm y_A}-r_0)^2 = R^2$ i.e., the disk $\mathcal{C}\left((0,0), R_{\rm g}\right)$, and the radar sector we then integrate $\rho (r)$ to find $l_{\rm B}$.

\textbf{Case 5:} This is similar to case 3, where the radar sector is partially inside the disk. Unlike case 3, here, the ego radar lies outside the disk. Thus, the circular arc of the radar sector lies inside the disk, and the remaining part lies outside. We find $l_{\rm B}$ using a similar analysis technique of case 3.

\end{document}